\documentclass[aps,floatfix,prb,superscriptaddress,twocolumn]{revtex4-2}
\usepackage{dcolumn}
\usepackage{graphicx}
\usepackage{amsmath, amssymb}
\usepackage{bm}
\usepackage{mathrsfs}
\usepackage{subfigure}
\usepackage{dsfont}
\usepackage{latexsym}
\usepackage{amsthm}
\usepackage{braket}
\usepackage{xcolor}

\def\Re{{\rm Re}}
\def\Im{{\rm Im}}
\def\Res{{\rm Res}}
\def\be{\begin{equation}}       \def\ee{\end{equation}}
\def\bea{\begin{eqnarray}}      \def\eea{\end{eqnarray}}

\newtheorem{theorem}{Theorem}
\newtheorem{corollary}{Corollary}
\newtheorem{lemma}{Lemma}

\newcommand{\ua}{\uparrow} 
\newcommand{\da}{\downarrow}

\newcommand{\ep}{\epsilon}

\newcommand{\bk}{\boldsymbol k}

\newcommand{\bkp}{\boldsymbol k^\prime}

\newcommand{\br}{\boldsymbol r}

\usepackage{times}

\begin{document}
\title{Non-Hermitian Chiral Superfluids with a Complex Interaction}

\author{Jia-Hang Ji}
\affiliation{College of Physics, Sichuan University, Chengdu 610064, China}

\author{Wenxing Nie}
\email{wxnie@scu.edu.cn}
\affiliation{College of Physics, Sichuan University, Chengdu 610064, China}

\date{\today}

\begin{abstract}

Recently, the influence of dissipation on a quantum system has attracted much attention, particularly on how the non-Hermitian terms modify the energy spectrum, band topology, and phase transition point. Motivated by the recent investigation of non-Hermitian $s$-wave superfluidity, we study the non-Hermitian chiral $p+ip$ superfluid (SF) with a complex-valued interaction, originating from inelastic scattering between fermions. We reformulate the non-Hermitian mean-field theory for chiral SFs and derive the gap equation in the path integral approach. By numerically solving the gap equation, we obtain the phase diagram of the non-Hermitian $p+ip$ SF, characterized by the reentrant SF transition and dissipation-induced SF phase, as a result of the evolution of the exceptional lines.  The method can be extended to higher partial-wave chiral SFs, such as $d+id$ and $f+if$-wave SFs.
We further consider such a chiral $p+ip$ SF on a square lattice, to investigate the influence of dissipation on topology. We find that the non-Hermitian skin effect is absent in the specific cylinder geometry, in which the topology associated with the edge modes and Chern number is robust to dissipation. Besides, we find that the energies at the robust point nodes and line nodes are pure real. We further verify the conditions of zero winding number in (quasi-) one-dimensional systems, and prove an associated ``no-go'' theorem, which is hopefully applied to explore the geometry dependent skin effect.

\end{abstract}

\maketitle
\section{Introduction}
\label{sec: Introduction}

Chiral superfluids (SFs), whose Cooper pairs carry nonzero relative orbital angular momentum, e.g. $ L_z=\nu\hbar$~\cite {Anderson1961,Leggett1975-RMP,Volovik-book,Leggett-book}, where $\nu=1,2,3$ for chiral $p, d$, and $f$-wave SFs, are featured by time-reversal symmetry breaking~\cite{Schnyder2008PRB,Chiu2016RMP}. Such chirality is encoded in the gap function, $\hat \Delta \sim (\hat p_x + i \hat p_y)^{\nu}$, where $\nu$ is odd for a spin triplet or even for a spin singlet. With weak pairing, in the Bardeen-Cooper-Schrieffer (BCS) regime, chiral SFs support an integer topological invariant~\cite{Volovik-JETP1988}, the Chern number~\cite{TKNN}, which coincides with the Cooper pair orbital angular momentum. According to bulk-edge correspondence, the same number of edge modes as the bulk topological invariant, like Chern number, which emerges at the boundary of the system, is supposed to be insensitive to weak perturbation due to topological protection~\cite{Hatsugai1993,Halperin1982,graf2013bulk,graf2018bulk}.

Although the notion of topological superconductor/superfluid is based on a Hermitian Hamiltonian, when the system inevitably couples to the environment, the effective Hamiltonian is a non-Hermitian (NH) one~\cite{daley2014quantum,dalibard1992wave}. Recently, non-Hermitian open quantum systems have been investigated extensively~\cite{burst2022,Abbasi2022,yongchun2022,yao2019PRL,Leykam2017,Gong2018PRX,Kawabata2019PRX,Ashida2020,Bergholtz2021RMP,Yao2018PRL1D,Yao2018PRL2D,Kunst2018PRL,Lee2019PRB,Helbig2020NatPhys,Xue2020NatPhys,Titus2023PRL,Titus2024PRL, Kawabata2020higher,Shen2018PRL,YokomizoPRL2019,Yokomizo2020non,Kawabata2020higher,Yang2020PRL,Kawabata2020PRB,Fang2022NatComm,Yokomizo2023,Amoeba2024PRX}, not only due to its significant topological properties (see the reviews~\cite{Gong2018PRX,Kawabata2019PRX,Ashida2020,Bergholtz2021RMP} and references therein), but also due to the generic properties unique to non-Hermitian systems, such as non-Hermitian skin effect~\cite{Yao2018PRL1D,Yao2018PRL2D,Kunst2018PRL,Lee2019PRB,Helbig2020NatPhys,Xue2020NatPhys,Titus2023PRL,Titus2024PRL, Kawabata2020higher} and non-Bloch band theory~\cite{Shen2018PRL,Yao2018PRL1D,YokomizoPRL2019,Yokomizo2020non,Kawabata2020higher,Yang2020PRL,Kawabata2020PRB,Fang2022NatComm,Yokomizo2023,Amoeba2024PRX}. When dissipation is introduced into strongly correlated systems, more intriguing phenomena emerge~\cite{Bouganne2020NatPhys, Chen2025NatPhys, f-sum-PRL2024}.  
Recent studies on the non-Hermitian fermionic Hubbard model with a complex-valued $s$-wave contact interaction~\cite{yamamoto2019theory,Iskin2021PRA}, originating from inelastic scattering, show a unique non-Hermitian effect: reentrant SF transition and dissipation-enhanced SF phase. 
This model is featured by the singularity in the superconducting gap due to Fermi surface instability, which has been found directly related to the number of the roots of the partition function, known as Lee-Yang zeros~\cite{YangLee1,YangLee2}, of a BCS superconductor~\cite{li2023yang}.

Non-Hermitian skin effect (NHSE)~\cite{Yao2018PRL1D,Yao2018PRL2D}, is used to describe the unique phenomenon in some non-Hermitian systems that a large number of bulk modes localize at the boundary of the system under the open boundary conditions (OBC), with exponentially decayed wavefunctions away from the edge. It is significantly different from Hermitian systems, whose wavefunctions are extended Bloch waves. The emergence of NHSE also indicates that the traditional Bloch band theory is no longer valid in such systems, and the traditional bulk-boundary correspondence breaks down, so we need to resort to non-Hermitian band theory with generalized Brillouin zone~\cite{Yao2018PRL1D,Yao2018PRL2D,YokomizoPRL2019}. It is ideal to predict the presence of non-Hermitian skin modes by calculating winding number on the complex energy plane~\cite{Fang2020PRL, Okuma2020, Fang2022NatComm}.

Previous studies show some interesting phenomena in non-Hermitian (NH) superfluids~\cite{yamamoto2019theory, Iskin2021PRA,Tajima2023PRA}, NH unconventional superconductors~\cite{Das2018PRB,Kou2023arXiv,Liuwuming2022}, and NH higher-order topological superconductors~\cite{XiangJi2024,Ghosh2022PRB}. However, when the non-Hermitian pairing caused by a complex-valued interaction is introduced to chiral SFs, which belong to the topological superconductors, many interesting questions arise, such as how dissipation influences the phase transition point, how robust the edge state is by topological protection, and whether there is the skin effect or not. To address these questions, we extend the non-Hermitian pairing to chiral SFs, study the phase diagram with respect to dissipation and the topological behaviors of the system on a two-dimensional square lattice.

The rest of the paper is organized as follows. In Sec.~\ref {sec: NHMF}, we first generalize the non-Hermitian pairing to chiral SFs, derive the non-Hermitian mean-field theory for the gap equation, and self-consistently solve the gap equation to obtain the phase diagram. In Sec.~\ref {sec: topological}, by using the self-consistently solved gap order parameter, we consider a two-dimensional chiral superconductor on a square lattice, under the periodic boundary conditions, and in a cylinder geometry with an open boundary, respectively, to investigate the influence of dissipation on energy spectra, edge states, and topology. In Sec.~\ref {sec: theorem}, we present a theorem for zero winding number in (quasi-) one-dimensional systems, and a corollary that applies to the model we consider. Finally, we present the conclusions and discussions in Sec.~\ref {sec: conclusion}.

\section{Non-Hermitian mean-field theory for chiral superfluids} 
\label{sec: NHMF}

To describe chiral SFs, we consider the relevant scattering process of particle pairs ($|\bk,\ua\rangle, |-\bk,\da\rangle$) with vanishing total momentum, in the form of a reduced Hamiltonian, 
\begin{equation}
H=\sum_{\bk,\sigma}\xi_{\bk}c^\dagger_{\bk,\sigma}c_{\bk,\sigma}
-\frac{1}{N}\sum_{\bk, {\bk}^\prime} V(\bk, {\bk}^\prime) c^\dagger_{\bk,\ua}c^\dagger_{-\bk,\da}c_{-\bk^\prime,\da}c_{\bk^\prime,\ua},
\label{eq:H_bcs}
\end{equation}
where $N$ is the number of lattice sites, and $c^\dagger_{\bk,\sigma}$ ($c_{\bk,\sigma}$) creates (annihilates) one fermion with momentum $\bk$ and spin $\sigma$. The first term in Eq.~\eqref{eq:H_bcs} is the kinetic energy measured relative to the chemical potential $\mu$, where $\xi_{\bk}=\epsilon _{\bk}-\mu$ with dispersion $\ep_{\bk}$. The second term describes the two-particle momentum-dependent interaction by a scattering matrix element $V(\bk,\bkp)$, with vanishing total momentum. For $s$-wave isotropic pairing, the two-particle interaction is reduced to a structureless scattering matrix element $g$, which is the bare coupling constant.

The scattering potential $V(\bk,\bkp)$ is a only function of $|\bk-\bkp|$. At the Fermi level, $k=k'=k_F$ can be taken. Then $V(\bk,\bkp)$ only depends on $\cos \theta _{\bk \bkp} $, where $\theta _{\bk \bkp}$ is the relative angle between $\bk$ and $\bkp$ fermions. Therefore,  the potential factor $V(\bk,\bkp)$ can be expanded by the Legendre function in the approach of partial wave expansion, $V(\bk,\bkp) =\sum_{l}(2l+1)V_{l} {\rm P}_{l}(\cos\theta _{\bk \bkp})$, where $l$ is the azimuthal quantum number. After applying the spherical harmonic addition theorem ${\rm P}_{l}(\cos\theta _{\bk \bkp})=4\pi/(2l+1)\sum_{m=-l}^{l} {\rm Y}_{lm}(\theta_{\bk},\phi_{\bk}) {\rm Y}_{lm}^{*}(\theta_{{\bkp}},\phi_{{\bkp}})$, we obtain:
\begin{equation}
V(\bk,\bkp)=4\pi\sum_{l}\sum_{m=-l}^{l}V_{l} {\rm Y}_{lm}(\theta_{\bk},\phi_{\bk}) {\rm Y}_{lm}^{*}(\theta_{{\bkp}},\phi_{{\bkp}}),
\label{eq:Vkk'}
\end{equation}
where ${\rm Y}_{lm}(\theta_{\bk},\phi_{\bk}) $ is the spherical harmonic function,  and $m$ is the magnetic quantum number as $m=-l,-l+1,\cdots,l$. 

The $p$-wave pairing is encoded in the component with $l=1$, $m=-1,0,1$. To describe the equal-spin pairing (ESP) state, the so-called Anderson-Brinkman-Morel phase (A-phase) of $^3$He or the chiral $p$-wave superfluids~\cite{Anderson1961,Leggett1975-RMP}, we can take $l=1$ and $m=+1$, and obtain the effective Hamiltonian,
\begin{equation}
H_{\text{p+ip}}=\sum_{\bk,\sigma}\xi_{\bk}c^\dagger_{\bk,\sigma}c_{\bk,\sigma}
- \frac{V_1}{N} \sum_{\bk, \bkp} \Gamma_{\bk} \Gamma_{\bkp}^{*} c^\dagger_{\bk,\ua}c^\dagger_{-\bk,\da}c_{-\bkp,\da}c_{\bkp,\ua},
\label{eq:H_p}
\end{equation}
where the interaction potential is assumed to be separable and is factorized as $V(\bk,\bkp)=V_1 \Gamma_{\bk}\Gamma^*_{\bkp}$ with $\Gamma_{\bk}=\sqrt{4\pi} {\rm Y}_{1,1}(\theta_{\bk},\phi_{\bk})$, encrypting the internal pairing of Cooper pairs.  The effective Hamiltonian for higher orders of pairing can be derived in the same way by considering higher channels $(l=2,3, \cdots)$.

\subsection{Non-Hermitian BCS theory for chiral superfluids}
\label{subsec: MF}

To describe the two-body loss process in NH BCS theory, we consider the Hamiltonian in Eq.~\eqref{eq:H_p}  with a complex-valued interaction $V_{1}=U_{1}+i\gamma/2$, where $U_{1}, \gamma>0$. The imaginary part of $V_1$ denotes the dissipation strength with $\gamma$.

We can define the order parameters for the superconducting phase as,
\begin{eqnarray}
\Delta&=&-\frac{V_{1}}{N}\sum_{\bk}\Gamma^*_{\bk} {}_L  \langle c_{-\bk,\da} c_{\bk,\ua}\rangle_{R},\label{eq:Delta}\\ 
\bar{\Delta}&=&-\frac{V_{1}}{N}\sum_{\bk} \Gamma_{\bk} {}_L \langle c^{\dagger}_{\bk,\ua} c^{\dagger}_{-\bk,\da} \rangle_{R}, \label{eq:Deltabar}
\end{eqnarray}
implying $\bar{\Delta} \ne \Delta^*$ due to $|E_n\rangle_R \ne |  E_n \rangle_L$ and $V_{1}\in \mathbb{C}$.
The mean-field approximation can be applied to the NH BCS Hamiltonian of the chiral $p+ip$ superfluid, by taking the substitute $c_{-\bk,\da} c_{\bk,\ua}={}_L \langle c_{-\bk,\da} c_{\bk,\ua}\rangle_{R}+\delta$ and neglecting the second-order of fluctuation $\delta$, yielding the effective mean-field Hamiltonian for the $p+ip$ chiral superfluid,
\begin{equation}
H_{\text{MF}}=\sum_{\bk,\sigma}\xi_{\bk}c^\dagger_{\bk,\sigma}c_{\bk,\sigma}+\sum_{\bk}\Big(\Delta \Gamma_{\bk}c^\dagger_{\bk,\ua}c^\dagger_{-\bk,\da}
+\bar{\Delta} \Gamma^*_{\bk} c_{-\bk,\da}c_{\bk,\ua}\Big).
\label{eq:Hmf}
\end{equation} 
We can rewrite it in the Nambu spinor basis,
\begin{equation}
H_{\text{MF}}=\sum_{\bk} 
(c^{\dagger}_{\bk,\ua} \quad c_{-\bk,\da})
\begin{pmatrix}
\xi_{\bk} & \Delta_{\bk}\\
\bar{\Delta}_{\bk} & -\xi_{\bk} 
\end{pmatrix}
\begin{pmatrix}
c_{\bk,\ua}\\ c^{\dagger}_{-\bk,\da}
\end{pmatrix},
\end{equation}
where $\Delta_{\bk}=\Gamma_{\bk}\Delta$ and $\bar{\Delta}_{\bk}=\Gamma_{\bk}^*\bar{\Delta}$. The effective Hamiltonian can be diagonalized by the Bogoliubov transformation,
\begin{equation}
\begin{pmatrix}
\gamma_{\bk,\ua} \\ \bar{\gamma}_{-\bk,\da}
\end{pmatrix}
=\begin{pmatrix}
u_{\bk} &  -v_{\bk}\\
\bar{v}_{\bk} &  u_{\bk}
\end{pmatrix}
\begin{pmatrix}
c_{\bk,\ua} \\ c_{-\bk,\da}^{\dagger}
\end{pmatrix}
\end{equation}
where $\gamma_{\bk,\sigma}$ and $\bar{\gamma}_{\bk,\sigma}$ are the quasiparticle operators, which are not Hermitian conjugate to each other due to the non-Hermiticity of the Hamiltonian. Therefore, the Hamiltonian is diagonalized by a similarity transformation but not a unitary one.  The BCS ground states can be constructed as $|\text{BCS}\rangle_{R}=\prod_{\bk}(u_{\bk}+v_{\bk}c^{\dagger}_{\bk,\ua}c^{\dagger}_{-\bk,\da})| 0\rangle$ for the right one, and ${}_L \langle\text{BCS}|= \langle 0 |\prod_{\bk}(u_{\bk}+\bar{v}_{\bk}c_{-\bk,\da}c_{\bk,\ua})$ for the left one. Here $|0\rangle$ is the vacuum state, and the coefficients satisfy $u^2_{\bk}+v_{\bk}\bar{v}_{\bk}=1$, as imposed by the normalized condition of the inner product ${}_L\langle \text{BCS}|\text{BCS} \rangle_{R}=1$.
Here, except for node points, where the system reduced to the Hermitian case, $v_{\bk}^{*}\neq\bar{v}_{\bk}$ in general, and $u_{\bk}\in\mathbb{R}$, $ v_{\bk}, \bar{v}_{\bk} \in\mathbb{C}$ by choosing a special gauge, 
\begin{subequations}
\begin{align}
u_{\bk}&=\sqrt{\frac{E_{\bk}+\xi_{\bk}}{2E_{\bk}}},\\
v_{\bk}&=-\frac{\Delta_{\bk}}{\sqrt{2E_{\bk}(E_{\bk}+\xi_{\bk})}},\\ 
\bar{v}_{\bk}&=-\frac{\bar{\Delta}_{\bk}}{\sqrt{2E_{\bk}(E_{\bk}+\xi_{\bk})}},
\end{align}
\end{subequations}
where the quasiparticle energy $E_{\bk}=\sqrt{\xi_{\bk}^2+\Delta_{\bk}\bar{\Delta}_{\bk}}$ is complex-valued in general. 

\subsection{Gap equation and quantum phase transition of the NH chiral $p$-wave superfluids}
\label{subsec: GapEq}

The gap equation of the NH chiral $p$-wave SFs can be derived in the path integral representation, where the partition function $\mathcal{Z}=\int \mathcal{D}[\bar{c},c] e^{S(\bar{c},c)}$. We introduce auxiliary bosonic fields $\bar{\Delta}$, $\Delta$ to perform the Hubbard-Stratonovich transformation. After integrating out the fermionic degrees of freedom, we obtain the partition function as $ \mathcal{Z}=\int\mathcal{D}[\bar{\Delta}, \Delta]e^{-S_{\text{eff}}}$. The effective action is given by (see Appendix~\ref{appA} for details),
\begin{equation}
S_{\text{eff}}(\bar{\Delta},\Delta)=-\sum_{\bk,\omega_n}\ln(\omega_{n}^{2}+\xi_{\bk}^{2}+\Delta_{\bk}\bar{\Delta}_{\bk})+\frac{\beta N}{V_1}\bar{\Delta}\Delta,
\label{eq:Seff}
\end{equation}
where $\Delta_{\bk}=\Gamma_{\bk}\Delta$, $\bar{\Delta}_{\bk}=\Gamma_{\bk}^*\bar{\Delta}$, and $\omega_n=(2n+1)\pi/\beta$ is the Matsubara frequency of fermions. Instead of the inverse of the temperature of the system, the parameter $\beta$ is set to the infinite limit to formulate a path integral, for obtaining an effective ground state. Applying the saddle point condition for the effective action, $\partial S_{\text{eff}}/\partial \Delta= \partial S_{\text{eff}}/\partial \bar{\Delta}=0$, and summing over Matsubara frequencies, we obtain the gap equation for the NH chiral $p$-wave SFs (see Appendix~\ref{appA} for details),
\begin{equation}
\sum_{\bk}\frac{|\Gamma_{\bk}|^2}{2E_{\bk}}\tanh \frac{\beta E_{\bk}}{2}=\frac{N}{V_{1}},
\label{eq:gap_eq}
\end{equation}
where $E_{\bk}=\sqrt{\xi_{\bk}^{2}+\Delta_{\bk}\bar{\Delta}_{\bk}}$. For the $s$-wave pairing SF, after taking $l=0$ in Eq.~\eqref{eq:Vkk'}, the gap equation is reduced to $\sum_{\bk}\tanh\big(\beta\sqrt{\xi^{2}_{\bk}+\bar{\Delta}\Delta}/2\big)/\big(2\sqrt{\xi^{2}_{\bk}+\bar{\Delta}\Delta}\big)=N/V_1$, which is the NH gap equation for the $s$-wave SF~\cite{yamamoto2019theory}.

We can set $\Delta=\bar{\Delta}\in\mathbb{C} $ by choosing a special gauge and numerically solving the self-consistent gap equation in Eq.~\eqref{eq:gap_eq}, by taking $\beta\to\infty$. The self-consistent solutions of Eq.~\eqref{eq:gap_eq} are presented in Fig.~\ref{fig:fig1}, with typical values of $U_1/t$. We consider here a square lattice with energy dispersion $\epsilon_{\bk}=-2t(\cos k_x+\cos k_y)$, where $t$ is the hopping amplitude, and the chemical potential is taken as $\mu/t=-1$. 
For the Hermitian chiral $p$-wave superfluids, the Chern number $C=-1$ when $-4<\mu/t<0$ on square lattice~\cite{bühler2014majorana,Huangwen2015}. Although for a non-Hermitian chiral superfluid, the topological phase transition point may shift and need to be reexamined in terms of the generalized Brillouin zone (GBZ), we can largely focus on the chemical potential relatively far away from the transition point, where the skin effect is expected to be neglected. Specially, the skin effect is exactly absent in the geometry we consider next, as the investigation, which is presented in Sec .~\ref {sec: theorem}, shows.
Therefore, we focus on the discussion of the case with $\mu/t=-1$ in the following. 

\begin{figure}[ht]
\centering
\includegraphics[width=8.5cm]{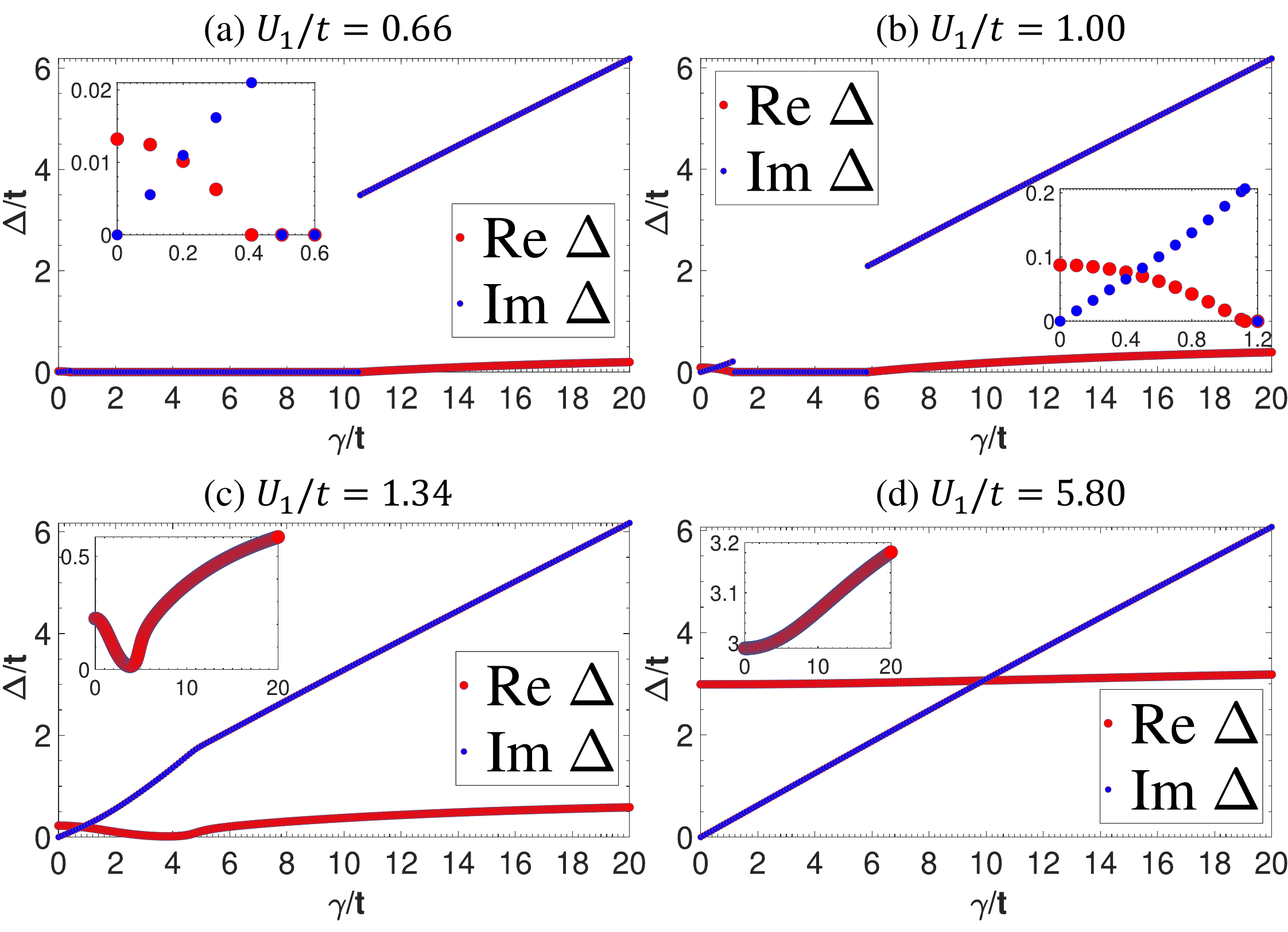}
\caption{The real part and imaginary part of the superfluid order parameter $\Delta$ as a function of $\gamma/t$, by numerically solved gap equation in Eq.~\eqref{eq:gap_eq}. Figures (a), (b), (c), and (d) show the solutions of $\Delta$ when $\mu/t=-1$, with $U_{1}/t=0.66, 1.00, 1.34$, and $5.80$, respectively, where the lattice size $N=1000\times1000$. }
\label{fig:fig1}
\end{figure}

From Fig.~\ref{fig:fig1} (a) and (b), we find that at weak pairing with $U_1/t=0.66$ and $U_1/t=1.00$, when the dissipation strength $\gamma$ increases, the nontrivial solution ${\rm Re} \Delta$ decreases first, then to zero, and increases later on, indicating the superfluid phase is destroyed first, and then is enhanced by dissipation. This two-fold role of dissipation is also observed in the non-Hermitian $s$-wave pairing SF. 
The insect of Fig.~\ref{fig:fig1}(a) includes two pure imaginary solutions of the gap equation at $\gamma/t=0.4079$ and $\gamma/t=10.5653$, and the insect of Fig.~\ref{fig:fig1}(b) includes two pure imaginary solutions at $\gamma/t=1.1227$ and $\gamma/t=5.8539$. These pure imaginary solutions, as critical points, locate at the phase transition boundary, which will be explained later in Fig.~\ref{fig:fig2}.
By comparing Fig.~\ref{fig:fig1}(a) and (b), we find the window for trivial zero solutions shrinks with increasing $U_1/t$, which is more obviously shown in Fig.~\ref{fig:fig2}. The critical point where the zero solution disappears is shown in Fig.~\ref{fig:fig1} (c) with intermediate pairing strength around $U_1/t=1.34$. At the strong pairing regime as shown in Fig.~\ref{fig:fig1} (d) with $U_1/t=5.80$, the superfluid order parameters ${\rm Re} \Delta$ and ${\rm Im} \Delta$ are both enhanced by any dissipation, instead of being suppressed first. 

To exclude the nontrivial solutions given by the possible local minimum of the real part of energy, we further compare the real parts of the energies of superfluid states and normal states by calculating the condensation energy, which is defined as the energy difference between the superfluid phase and normal phase. The condensation energy is derived as $E_{\text{c}}=N\bar{\Delta}\Delta/V_1-\sum_{\bk}(E_{\bk}-\lvert \xi_{\bk} \rvert)$. The details of the derivation are presented in Appendix~\ref{appB}.

\begin{figure}
\centering
\includegraphics[width=8.cm]{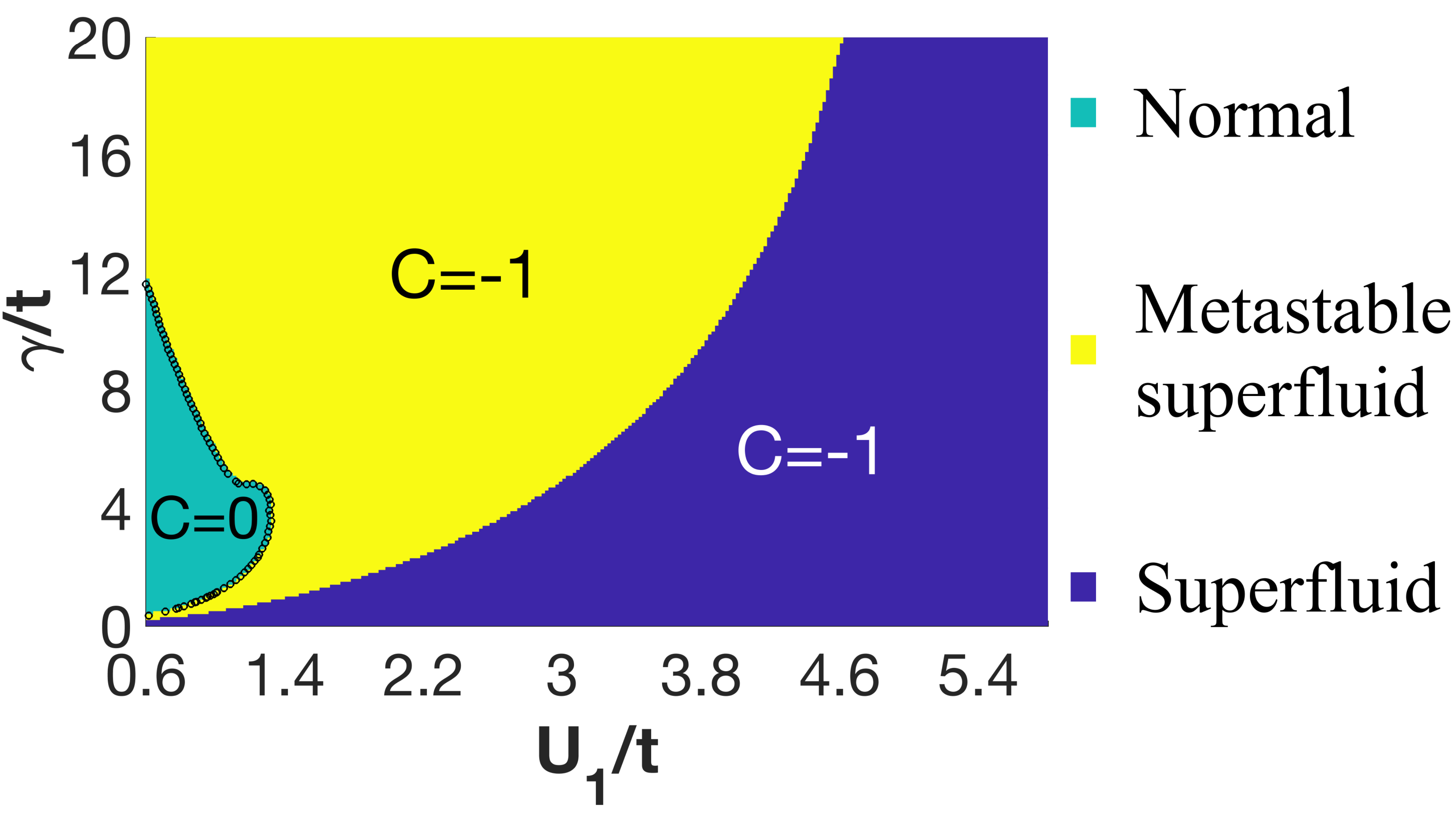}
\caption{Phase diagram of non-Hermitian $p+ip$ BCS model by self-consistently solved gap equation in Eq.~\eqref{eq:gap_eq} with $\mu/t=-1$, in the limit $\beta\to\infty$. The cyan region represents the normal phase without any convergent solution of the gap equation. The yellow region represents the metastable superfluid phase, with nonzero solutions and positive $\Re E_{\text c}$.  The purple region represents the stable superfluid phase, with nonzero solutions and negative $\Re E_{\text c}$. The values of the Chern number are labeled in these three phases. The black empty circles show the phase boundary between the normal phase and metastable superfluid phase, which is given by pure imaginary trial solutions of the gap equation.}
\label{fig:fig2}   
\end{figure}

After iterating Eq.~\eqref{eq:gap_eq} self-consistently for solutions of $\Re\Delta$ and $\Im\Delta$, and calculating the condensation energy, we construct the phase diagram of the BCS Hamiltonian with different values of $U_1/t$ and $\gamma/t$, as shown in Fig.~\ref{fig:fig2}. Following previous criteria, the diagram is characterized by three phases. The cyan region shows the ``normal phase'', in which the numerical calculation can not converge to any nonzero solution of $\Re\Delta$ or $\Im\Delta$. In the normal phase, the condensation energy is zero, as expected. The yellow and the purple regions correspond to the superfluid phase, in which the numerical calculation arrives at convergent solutions with finite $\Re\Delta$ and $\Im\Delta$ with reliable precision. The nonzero superconducting order parameter suggests a superfluid phase in such two regions. However, the sign of the real part of condensation energy $\Re E_\text{c}$ is positive (negative) in the yellow (purple) region. The positive $\Re E_\text{c}$ indicates a ``metastable superfluid phase'', attributed to the contribution of the non-Hermitian effect introduced by pairing. While the negative $\Re E_\text{c}$ indicates a ``stable superfluid phase'' as shown in the purple region, which is the effective ground state of the non-Hermitian BCS Hamiltonian. In the case $\gamma=0$, it is our familiar superfluid phase in the Hermitian case. 
The phase diagram is similar to that of the NH $s$-wave superfluid, but one can verify that any pure imaginary energy gap is a solution of the gap equation, which is different from the $s$-wave case, whose imaginary part has a upper bound at the phase transition boundary between normal phase and metastable superfluid phase. 

\begin{figure}
\centering     
\includegraphics[width=8.5cm]{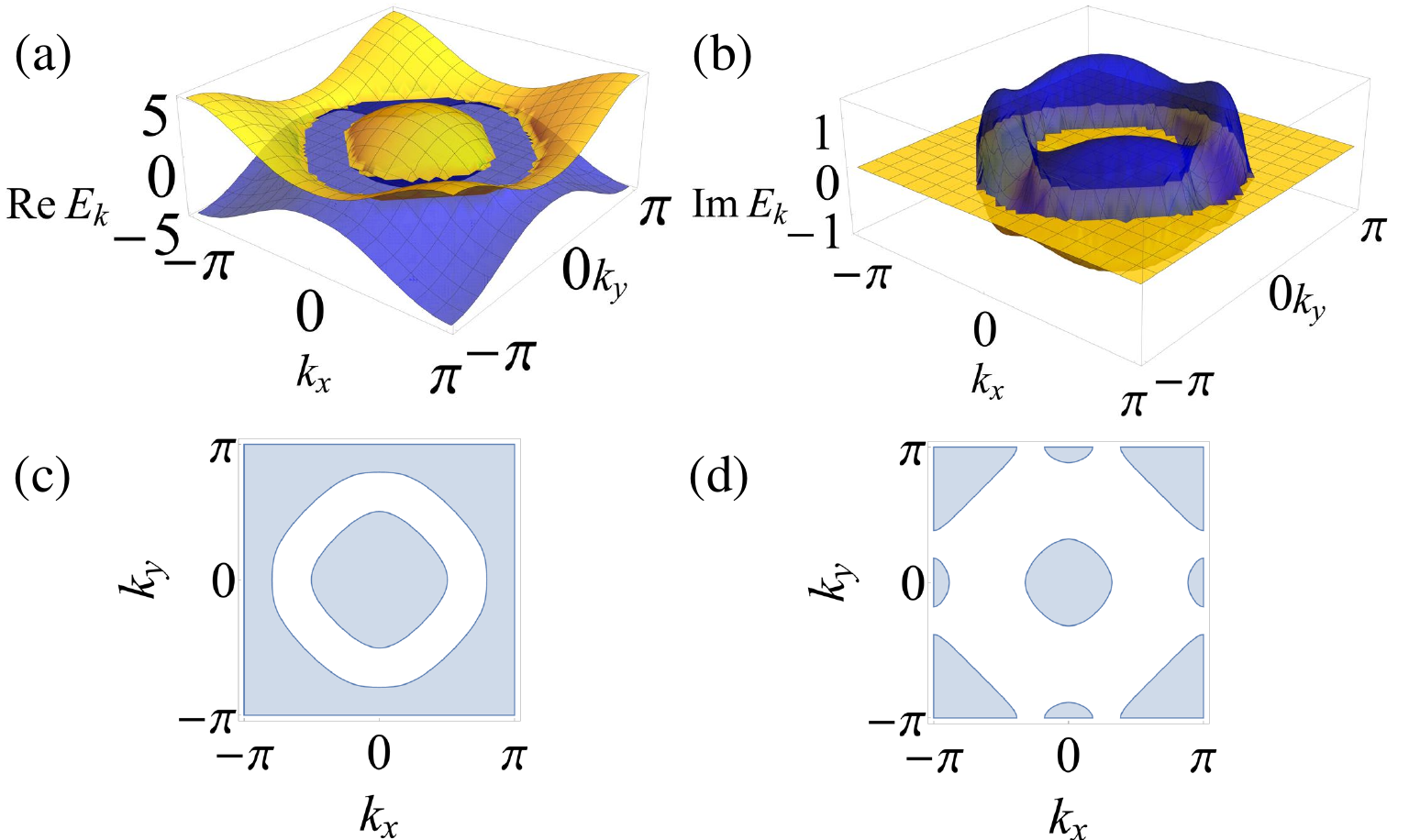}
\caption{(a) The real part and (b) imaginary part of quasiparticle energy spectra of the NH chiral $p$-wave SFs, $E_{\bk}=\sqrt{\xi_{\bk}^2+\Delta_{\bk}\bar{\Delta}_{\bk}}$, with $\mu/t=-1$, at critical point with $\Delta=0.8i$, with $\xi_{\bk}=-2t(\cos k_x+\cos k_y)-\mu$ and $\Delta_{\bk}=\Gamma_{\bk}\Delta$. The exception lines and two areas divided by them, when  $\mu/t=-1$,  with $\Delta=0.8i$ in (c), and $\Delta=2.0i$ in (d). Pure real energy states are labeled by the blue area, while pure imaginary energy states are labeled by the white area, indicating that the $CP$ symmetry spontaneously breaks at EPs.}
\label{fig:fig3}   
\end{figure}

We further investigate the phase boundary between the normal phase and the metastable superfluid phase, and find that it is a signature of the exceptional points of the non-Hermitian Hamiltonian. To show this, we take pure imaginary trial solution $\Delta=i \Delta_I, (\Delta_I\in\mathbb R)$, into the gap equation in Eq.~\eqref{eq:gap_eq}, and find corresponding $V_1$, whose real part $U_1$ and imaginary part $\gamma/2$ are labeled by black empty circles in Fig.~\ref{fig:fig2}. We find that these circles locate at the phase boundary between the normal phase and the metastable superfluid phase, which indicates that the pure imaginary solutions yield the phase boundary, which is also consistent with the analytical solution, as shown in the Appendix~\ref{appC}.

As one example with $\Delta=0.8i$ at the critical point, Fig.~\ref{fig:fig3} (a) and Fig.~\ref{fig:fig3} (b) show the real part and imaginary part of quasi-particle energy spectra in two dimensions, respectively, where the degenerate states known as exceptional points (EPs) are shown as dark blue solid lines in Fig.~\ref{fig:fig3} (c), accompanied by two areas divided by them. The blue area corresponds to the energy spectrum with pure real $E_{\bk}$, and the white area corresponds to the energy spectrum with pure imaginary  $E_{\bk}$. We also demonstrate the EPs with $\Delta=2i$ in Fig.~\ref{fig:fig3} (d), which shows the pure real area decreases as the absolute value of the pure imaginary gap increases. However, it is worth noting that the pure real area will not disappear because of the anisotropic nature of $p$-wave superconductivity, which is different from the behavior in the $s$-wave model.

We further verify that the parity-particle-hole ($CP$) symmetry~\cite{Okugawa2019,yamamoto2019theory,Budich2019,kawabata2019CP} holds in $H_{\text{MF}}(\bk)$ with pure imaginary pairing $\Delta$, $CP H_{\text{MF}}(\bk) (CP)^{-1}=-H_{\text{MF}}(\bk)$, where $CP=\sigma_x K$, $\sigma_x$ is Pauli matrix, and $K$ is complex conjugation. When the eigenstates degenerate, the $CP$ symmetry of $H_{\text{MF}}(\bk)$ spontaneously breaks at EPs in momentum space. Therefore, the phase transition between the normal phase and the metastable superfluid phase is completely induced by the non-Hermitian effect, accompanied by the emergence of EPs. 

\section{Topological behavior of the system}
\label{sec: topological}
Different from the $s$-wave SF, which is topologically trivial, the nontrivial topology of the chiral $p$-wave SFs is encoded in its order parameter. After introducing the non-Hermitian pairing, it turns into an interesting platform to investigate whether and how dissipation induced by the non-Hermitian pairing influences the topological properties of this system. 
In the following, we consider a two-dimensional chiral SF on a square lattice with the non-Hermitian pairing,
\begin{equation}
\begin{aligned}
H&=-t'\sum_{\br,\sigma}\Big(c^\dagger_{\br,\sigma}c_{\br \pm \hat{x},\sigma}+c^\dagger_{\br,\sigma}c_{\br\pm\hat{y},\sigma}\Big)-\mu '\sum_{\br,\sigma}c^\dagger_{\br,\sigma}c_{\br,\sigma}\\
&+\frac{g}{2}\sum_{\br,\zeta=\pm 1}\Big(i\zeta\Delta c^\dagger_{\br,\ua}c^\dagger_{\br+\zeta \hat{x},\da}-\zeta\Delta c^\dagger_{\br,\ua}c^\dagger_{\br+\zeta\hat{y},\da}\\
&-i\zeta\bar{\Delta}c_{\br+\zeta\hat{x},\da}c_{\br,\ua}-\zeta\bar{\Delta} c_{\br+\zeta\hat{y},\da}c_{\br,\ua}\Big),
\end{aligned}
\label{eq:H_lattice}
\end{equation}
where the chemical potential $\mu '$ is real, and $\bar{\Delta} \ne \Delta^*$, implying the non-Hermitian pairing. We let $t'=2t$, $\mu '=2\mu$, and $g=\sqrt{3/2}$ to guarantee the values of $\Delta$ and $\bar{\Delta}$ are the same as those in Eq.~\eqref{eq:Hmf}, so that we can use the numerical solutions of order parameters obtained by solving the gap equation in Eq.~\eqref{eq:gap_eq}.

\begin{figure}[hb]
\centering
\includegraphics[width=8.5cm]{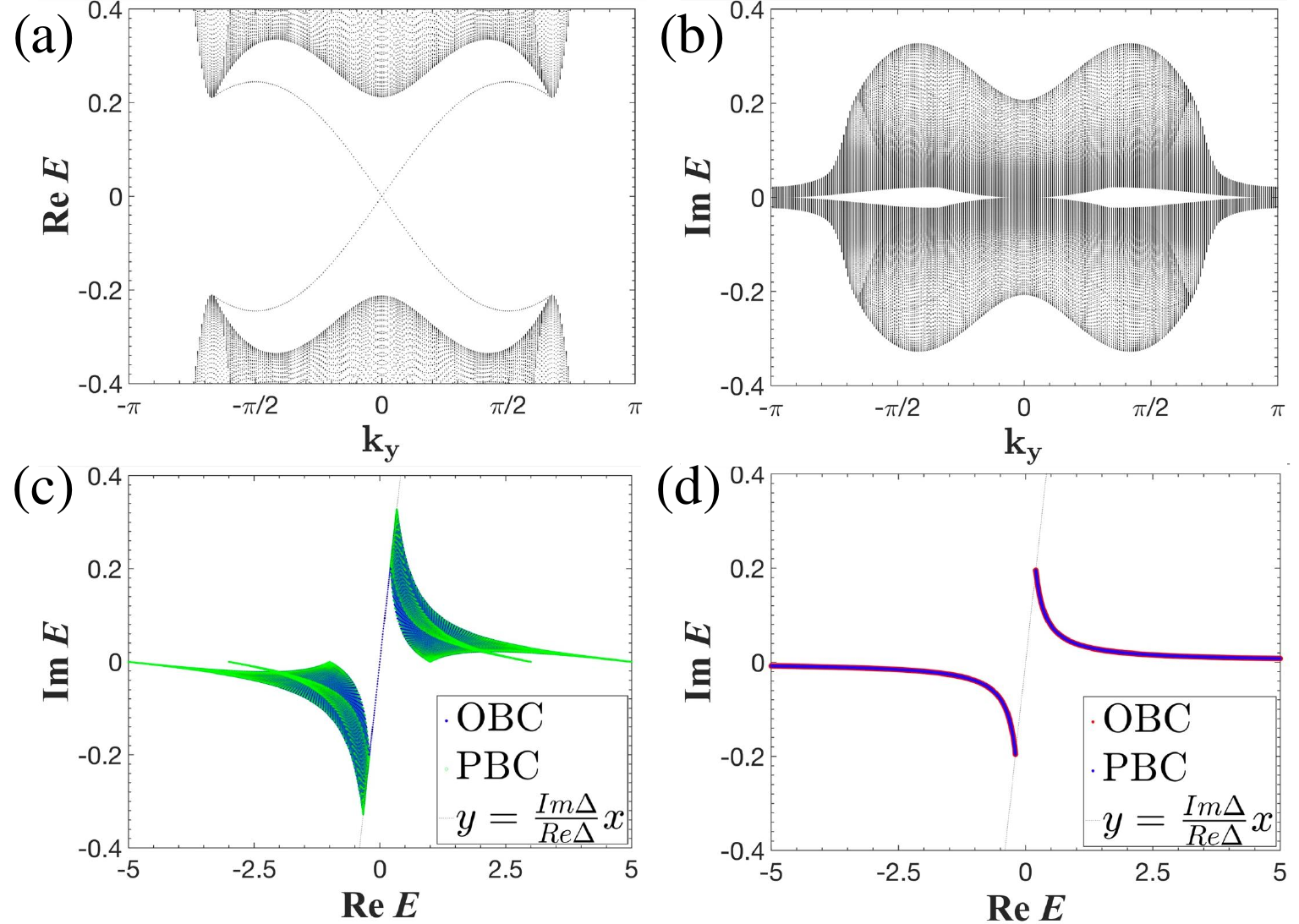}
\caption{The low energy dispersions of the 2D non-Hermitian chiral $p$-wave superfluids on a square lattice in the cylinder geometry with  periodic boundary conditions are shown in (a) and (b). The complex spectra of this system with PBC and in cylinder geometry with OBC are shown in (c), where the chiral edge state shown by blue dots can be linearly fitted by the function $\Im E=(\Im \Delta/ \Re \Delta)\Re E$.  Parameter in Fig.~(a)-(c): $\Delta/t=0.1998+0.1954i$, one stable superfluid solution of the gap equation, when $V_1/t=1.34+0.4i$. The cylinder height employed is $L_x=250$. (d) A schematic plot of the complex energy spectrum of the $s$-wave SF, where the area of the spectrum is zero, i.e. forms arcs on the complex plane.}
\label{fig:fig4}   
\end{figure}

With periodic boundary conditions (PBC) in both $x$ and $y$ directions, the Hamiltonian in Eq.~\eqref{eq:H_lattice} can be written in Bogoliubov-de Gennes (BdG) form as $H=\sum_{\bk} \psi_{\bk}^{\dag} H_{\bk} \psi_{\bk}$ in momentum space after the Fourier transformation, where $\psi_{\bk}$ is a Nambu spinor that $\psi_{\bk}=(c_{\bk,\ua},c_{-\bk,\da}^{\dag})^{T}$. Here, $H_{\bk}$ is a Bloch Hamiltonian,  
\begin{equation}
H_{\bk}=
\begin{pmatrix}
\xi_{\bk} & \Delta_{\bk}\\
\bar{\Delta}_{\bk} & -\xi_{\bk} \\
\end{pmatrix}
\label{eq:BlochH}
\end{equation}
where $\xi_{\bk}=-2t(\cos{k_x}+\cos{k_y})-\mu$, $\Delta_{\bk}=\sqrt{3/2} \Delta (\sin{k_x}+i \sin{k_y})$, $\bar{\Delta}_{\bk}=\sqrt{3/2} \bar{\Delta}(\sin{k_x}-i \sin{k_y})$. 
We also set $\Delta=\bar{\Delta}\in\mathbb{C} $ by choosing a special gauge. It is easy to check that this Bloch Hamiltonian satisfies the particle-hole symmetry (PHS): $PH_{\bk}^TP^{-1}=-H(-\bk)$, where $P=\sigma_x$ is one Pauli matrix. Hence, the non-Hermitian $p+ip$-wave SF belongs to the class D by 38-fold symmetry classification of non-Hermitian systems~\cite{Kawabata2019PRX}.

Figure~\ref{fig:fig4} shows the energy spectra of this system in cylinder geometry under the open boundary conditions (OBC) in $x$ direction, and in a torus geometry with PBC, respectively. 
We mainly focus on the discussion with $\mu/t=-1$, so that we can use the numerical solutions of the gap equation, obtained in the phase diagram as shown in Fig.~\ref{fig:fig2}. Figures~\ref{fig:fig4} (a) and (b) show the real part and the imaginary part of the low energy dispersion with OBC, respectively, when $\Delta/t=0.1998+0.195i$, which is a stable superfluid phase solution of the gap equation at weak interaction $V_1/t=1.34+0.4i$. As stated in the Ref.~{\cite{Kawabata2019PRX}, the chiral edge states are robust in the presence of a line gap, for non-Hermitian topological superconductors belonging to class D, due to the $\mathbb{Z}$ topology, and this is also the characteristic of the system that we are considering. Hence, the edge state is still robust to dissipation and well localized near the boundary of the cylinder, as shown in Fig.~\ref{fig:fig5} (b), although the value of the imaginary part of $\Delta$ is comparable with that of the real part. In the OBC spectrum on the complex energy plane as shown in Fig.~\ref{fig:fig4} (c), we can further fit the edge modes by a linear function $\Im E=(\Im \Delta/ \Re \Delta)\Re E$, whose slope is consistent with Ref.~\cite{Kou2023arXiv} when considering complex Fermi velocity. We have also rigorously proved that each PBC complex energy satisfies $\Im E/\Re E\le \Im \Delta/\Re \Delta$, and explained the folding-like shape of energy spectrum caused by this relation, as presented in Appendix~\ref{appD}. Namely, we have proved one upper bound of the ratio $\Im E/ \Re E$ in the PBC energy spectrum. 

Whereas, we still cannot ensure the robustness of the traditional bulk-edge correspondence, so we need to know whether the skin effect exists in the system or not. Different from the NH $s$-wave SF as illustrated in Fig.~\ref {fig:fig4} (d),  in which the spectral area is zero, where the skin effect can be proved to be absent~\cite{Fang2020PRL}, the PBC spectrum of the chiral $p$-wave superfluids in Fig.~\ref{fig:fig4} (c) covers a finite area on the complex plane, indicating that the universal skin effect can exist in a generic two-dimensional geometry with fully open boundaries~\cite{Fang2022NatComm}. However, we note that, except for the edge modes that appear as a straight line with OBC, the complex spectrum with OBC and the one with PBC totally overlap on this complex plane, implying that the possible relation ``GBZ=BZ", the condition of the absence of skin effect in the stripe geometry. 
In fact, in the specific geometry, as a cylinder with open boundary in $x$- or $y$-direction, we indeed prove the zero winding number for the NH chiral $p$-wave SFs by presenting Theorem~\ref{thm.theorem} and Corollary~\ref{corollary} in the next section, and the zero winding number is equivalent to the absence of the skin effect for these models (including the one considered here) shown in Appendix~\ref{appE}. Furthermore, this result also implies the absence of the skin effect for the NH chiral $p$-wave SFs in the square geometry under fully OBC, according to the geometry-dependent skin effect (GDSE)~\cite{Fang2022NatComm,zhou2023GDSE}. Therefore, we make the conclusion that the skin effect does not appear in this system, as verified in Fig.~\ref{fig:fig5}, where the distribution of wavefunctions of edge modes and bulk modes for an arbitrary $k_y$ was investigated. 

\begin{figure}[hb]
\centering
\includegraphics[width=7.5cm]{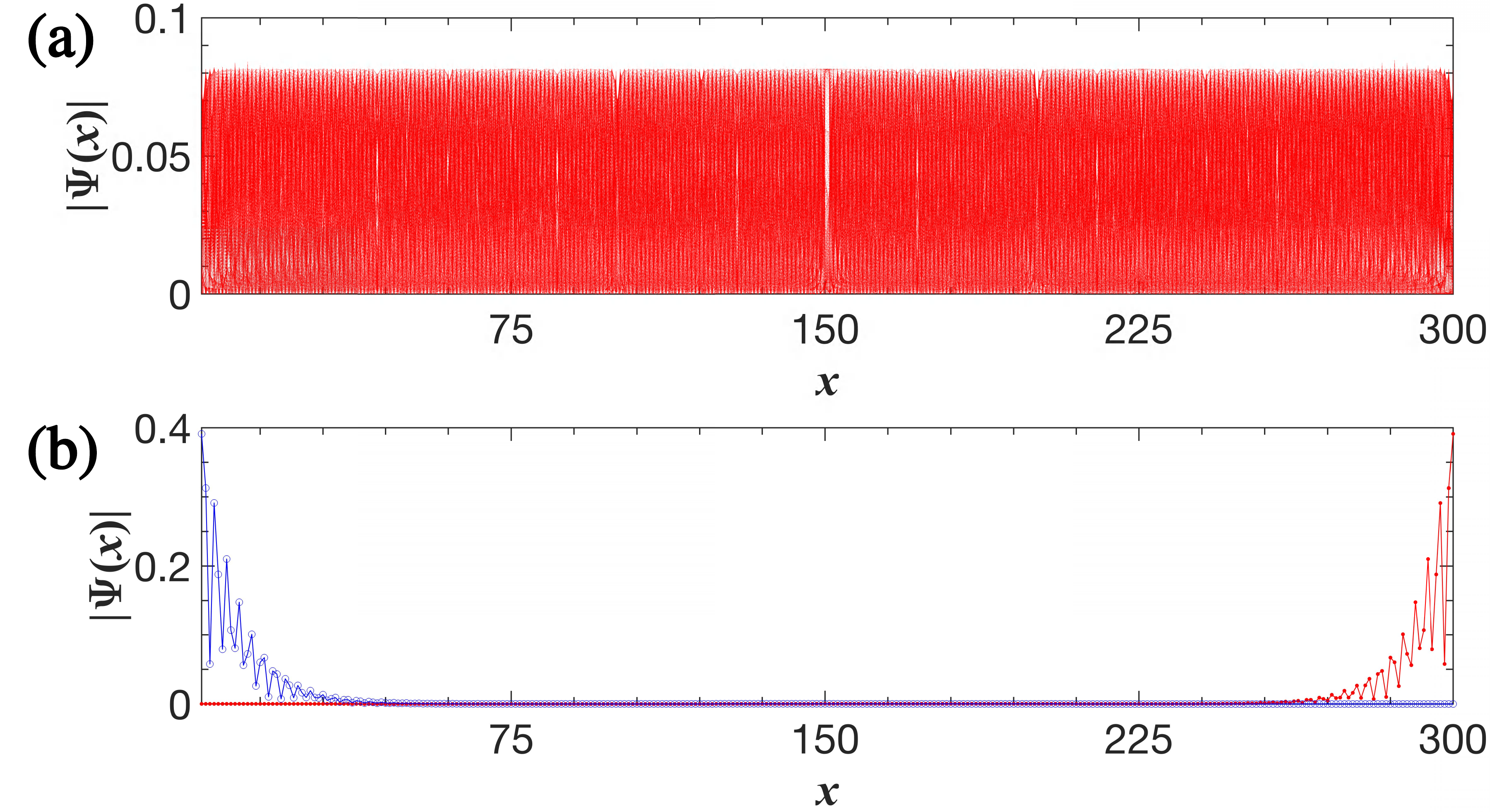}
\caption{The wavefunctions of all bulk modes in (a), and two edge modes characterized by different colors in (b), with $\mu/t=-1, U/t=1.34+0.4i$, and $\Delta=0.1998+0.1954i$. Here we fix an arbitrary $k_y=\pi/10$, and take system size $N=300$. }
\label{fig:fig5}   
\end{figure}

Another notable feature of the energy spectrum is the robustness of point nodes in the $p$-wave case. For Hermitian superconductivity, gap nodes can be defined as the $k$ points in the BZ that satisfy $|\Delta_k|^2=|\Delta \Gamma_k|^2=0$~\cite{sigrist2005introduction}. There is no node in the $s$-wave superconductor, because of its isotropic nature. However, for anisotropic unconventional superconductors, for example in the $p+ip$ chiral superfluid, there are two point nodes on the Fermi surface, one in the north pole and the other in the south pole~\cite{sigrist2005introduction}. The node corresponds to the gapless region in the Bogoliubov quasi-particle energy spectrum~\cite{ando2015topological}. For the NH $p$-wave superfluids, we find that the energy spectrum contains some points where the imaginary part of the energy is zero. We call this phenomenon as robustness of point nodes. The reason is explained in the following.

For a general non-Hermitian superconductor, the quasi-particle energy $E_k=\sqrt{\xi_k^2+\Delta^2|{\Gamma_k}^2|}$, where $\Delta\equiv(a+ib)\in \mathbb{C}$. So we have
\begin{equation}
\Re(E_k)\Im(E_k)=ab|\Gamma_k^2|.
\end{equation}
Suppose $a \ne 0$ and $b \ne 0$, it is obvious that if $\Im (E_k)=0$ at certain $k$ point, then $|\Gamma_k|^2=0$ there, implying the existence of nodes as such $k$ point. On the other hand, if $|\Gamma_k^2|=0$ at certain $k$ point, then we have $E_k=\sqrt{\xi_k^2}\in \mathbb{R}$, or $\Im (E_k)=0$ at nodes. Hence, we have the statement ``$\Im(E_k)=0 \Longleftrightarrow$ nodes in non-Hermitian superconductivity''. This statement shows that the dissipation will not change the value of $E_k$  at the nodes, as the BdG Hamiltonian changes from Hermitian one to a non-Hermitian one by a complex-valued interaction. This observation is not limited to point nodes. We also show the non-Hermitian $d_{x^2-y^2}$ superconductivity with line node as another example in the Appendix~\ref{appD}. Therefore, we can use dissipation as a probe to detect the type of nodes by measuring the energy spectrum.

\section{A theorem for the trivial winding number}
\label{sec: theorem}

The condition for non-Hermitian skin effect (NHSE) is an important issue.  We can use the generalized Brillouin zone (GBZ) condition and solve the characteristic equation~\cite{Yao2018PRL1D,YokomizoPRL2019}, or use a more general and accurate method, auxiliary GBZ~\cite{Yang2020PRL}, to predict its presence. These methods usually involve sophisticated calculations, especially in dimensions higher than one~\cite{Amoeba2024PRX}. 

It has been found that NHSE is related to intrinsic non-Hermitian topology~\cite{Fang2020PRL,Okuma2020,HU2025NHSE}. In particular, in (quasi-) one dimension systems, one can define a non-Hermitian winding number as a topological invariant~\cite{Gong2018PRX,Kawabata2019PRX,Fang2020PRL,Okuma2020}. 
The winding number of energy on the complex plane concerning a reference energy $E_0$ of a (quasi-) one-dimensional (1D) non-Hermitian model is defined as~\cite{Fang2020PRL,Gong2018PRX}:
\begin{equation}
w_{\mathcal C, E_0}:=\frac{1}{2\pi}\oint_{ \mathcal C} \frac{{\rm d}}{{\rm d}z}\arg[H(z)-E_0] {\rm d}z,
\label{eq:WN}
\end{equation}
where $\mathcal C$ is an arbitrary oriented loop, and $E_0\in\mathbb{C}$ represents an arbitrary reference point on the complex energy plane. If we choose $\mathcal C=\text{BZ}$ specially, 
Eq.~\eqref{eq:WN} describes the winding phase of $H(z)-E_0$ along BZ, which is a topological invariant unique in non-Hermitian systems~\cite{Gong2018PRX,Kawabata2019PRX,Shen2018PRL,Bergholtz2021RMP}.
We can let $z:=e^{ik}$, where the Brillouin zone (BZ) is mapped into a unit circle $|z|=1$, and the function $H(z)$ is the corresponding image of the Bloch Hamiltonian $H(k)$ via this mapping. The winding number reflects whether the image of the unit circle under $H(z)$ is an arc with zero interior or a loop with nonzero area enclosed with PBC.
 It can be rewritten as an integral in momentum space:
\begin{equation}
w_{\text{BZ}, E_{0}}=\frac{1}{2\pi i}\int_{0}^{2\pi}\frac{{\rm d}}{{\rm d}k}\ln\Big[\det[H(k)-E_{0}\mathbb{I}]\Big]{\rm d}k.
\label{eq:windingBZ}
\end{equation}

It has been theoretically established the exact correspondence between the emergence of non-Hermitian skin modes with open boundary condition, and the nonzero winding numbers $w_{\text{BZ}, E_{0}}$ on the complex energy plane as momentum $k$ traverses BZ with PBC, for (quasi-) one-dimensional one-band systems~\cite{Fang2020PRL,Borgnia2020,Okuma2020}. 
If the image of the unit circle $H(z)$ with PBC is loop-shaped, NHSE is proved to be present. On the contrary, if the image of the unit circle only forms an arc without interior, NHSE is proved to be absent~\cite{Fang2020PRL}.
As the winding number represents the number of times that the energy spectrum wraps around $E_0$ with OBC, once we get a nonzero winding number $w_{\text{BZ}, E_{0}}$ around a  certain $E_0$ inside the PBC curve, the spectrum must be a loop, implying the GBZ must deviate from BZ, since $w_{\text{GBZ}, E_{0}}=0$ always holds~\cite{Fang2020PRL}. Hence, evaluating winding number along BZ is important and useful to predict the skin effect in (quasi-) one-dimensional systems.

It would be very convenient if we could find a simple criterion to determine whether the winding number is zero, rather than specific calculations for it. The winding number of a class of Hamiltonians with certain symmetries~\cite{Gong2018PRX,Liu2019,Yi2022} or for disordered system~\cite{Claes2021}, based on the analysis of Hamiltonians, has been revealed by some previous studies. 
In the following, we will provide a ``no-go” theorem as a criterion by considering the algebraic form of a Hamiltonian, which reflects the mathematical structure behind the winding number. 

\begin{theorem}
For an $E_0\in\mathbb{C}$, if the characteristic polynomial $F(k)\equiv \det[H(k)-E_{0}\mathbb{I}]=\sum_{l,m,n,a,b} c_{l}\cos^m(ak)\sin^n(bk)$  is an even function of $k$, then we have $w_{\text{BZ}, E_0}=0$, where $l,m,n,a,b\in\mathbb{N}$ and $c_{l}\in\mathbb{C}$.
\label{thm.theorem}
\end{theorem}

The proof of Theorem~\ref{thm.theorem} is presented in the Appendix~\ref {appE}. As an application of Theorem~\ref{thm.theorem}, we present the following corollary for a widely used two-band Hamiltonian in Pauli matrix form.

\begin{corollary}
For a two-band Hamiltonian $H(k)=a_0(k)\sigma_0+a_1(k)\sigma_x+a_2(k)\sigma_y+a_3(k)\sigma_z$. If $a_0(k)$ and $\det[H(k)]$ are both even, then we have $w_{\text{BZ}, E_0}=0$ for any arbitrary $E_0\in\mathbb{C}$.
\label{corollary}
\end{corollary}

According to Theorem~\ref{thm.theorem}, Corollary~\ref{corollary} can be easily proved when one notices that $F(k)= \det[H(k)-E_{0}\mathbb{I}]=\det[H(k)]-2E_0a_0(k)+E_0^2$.

Theorem~\ref{thm.theorem} and Corollary~\ref{corollary} provide a convenient way to determine zero winding number in a broad class of (quasi-) 1D non-Hermitian models. For example, for a 1D Hamiltonian with reflection symmetry, $RH(-k)=H(k)R$, where $R$ is a unitary operator, we can find that $\det[H(k)-E\mathbb{I}]=\det[R\big(H(-k)-E\mathbb{I}\big)R^{-1}]=\det[H(-k)-E\mathbb{I}]$. Since it satisfies the condition of Theorem~\ref{thm.theorem}, we can conclude that a system with reflection symmetry has a trivial winding number. Further, if its topological number is characterized by the winding number, i.e. for all $\mathbb{Z}$ classes, we can conclude that there is no skin effect~\cite{Kawabata2019PRX,Fang2020PRL},  like the model of class AIII~\cite{Liu2019,Gong2018PRX}. It would be worthwhile to make similar attempts for other possible symmetries to investigate the skin effect, as a method independent of the Ref.~\cite{Yi2022}, where the skin effect is determined by the GBZ condition for systems with characteristic equations constrained by specific symmetries.
The application of Theorem~\ref{thm.theorem} to a non-Hermitian Kitaev chain Hamiltonian is presented in Appendix~\ref{appE}. 

Besides, since higher-dimensional skin effect may be related to 1D winding number~\cite{Fang2022NatComm}, Theorem~\ref{thm.theorem} and Corollary~\ref{corollary} could give some insight for higher dimensions. In principle, one can define the spectral winding number for a straight line in the BZ~\cite{Fang2022NatComm}, which requires us to consider all possible reference energies on the complex energy plane, corresponding to entire BZ, for obtaining the information of skin effect on the stripe geometry. However, in the specific calculation, we may have a divergence situation, as illustrated in the Appendix~\ref{appE}. Despite this, once the Hamiltonian considered satisfies the conditions of Corollary~\ref{corollary}, we can immediately conclude that the winding number is zero for all reference energies, which guides us to the zero spectral winding number of the corresponding straight line in the higher-dimensional BZ, and further, the possible absence of skin effect on the corresponding stripe geometry. This conclusion could be beneficial for finding edge directions that do not exhibit the NHSE in systems with GDSE~\cite{Fang2022NatComm}, and further predicting sample geometries that might exhibit the skin effect, such as some non-Hermitian phononic crystals with exceptional points~\cite{zhou2023GDSE}.

In particular, for the NH chiral $p$-wave SFs considered here, we fix one $k_x$ ($k_y$) in the 2D BZ, and calculate the spectral winding number of the vertical line (horizontal line) for all reference energies on the complex plane. After repeating the procedure for all  $k_x$ ($k_y$), we will know the information of the winding number in the cylinder geometry with open boundaries. However, this is a two-band model for a fixed $k_y$ or $k_x$, and one should notice that the trivial one-dimensional winding number may not be equivalent to the absence of the skin effect in multi-band systems~\cite{guo2023anomalous,Fang2020PRL}. Despite this, we can turn our model to a single-band problem, as stated in the Appendix~\ref {appE}, and zero winding number could indicate the absence of the skin effect~\cite{Fang2020PRL}. Hence, Corollary~\ref{corollary} shows that the cylinder geometry with either arbitrary vertical boundaries or horizontal boundaries as the model we employed has a zero winding number, and the system will also not show NHSE in a square sample. But in general, once we change the shape of samples under fully OBC, we might find NHSE appears. This is a typical example of GDSE. Namely, even if the area of the energy spectrum is nonzero, the skin effect disappears in certain fully open boundary geometries~\cite{Fang2022NatComm}. The possible emergence of NHSE also means that we cannot use the traditional BZ to calculate the Chern number in a generic geometry, implying that dissipation potentially influences the topological properties of chiral SFs.

Since the system belongs to class D, and it is proved that, by Theorem~\ref{thm.theorem}, there is no skin effect for the non-Hermitian $p$-wave SFs we considered in Sec~\ref{sec: topological}, we can calculate the Chern number by its definition in BZ as well,
\begin{equation}
C=\frac{1}{2\pi}\int_{BZ}{\rm d}k_{x}{\rm d}k_{y}[\partial_{x}a_{y}(k_{x},k_{y})-\partial_{y}a_{x}(k_{x},k_{y})],
\end{equation}
where $a_{j}(k_{x},k_{y})$ is the Berry connection, $a_{j}(k_{x},k_{y})=-i_{R}\bra{\bold{k}} \partial_{j} \ket{\bold{k}}_{R}$
with $j=x,y$.
The results for different values of the complex interaction are labeled in Fig.~\ref{fig:fig2}, which verify the survival of the traditional bulk-edge correspondence again.

\section{conclusion}
\label{sec: conclusion}

In this manuscript, we have reformulated the mean-field theory for the non-Hermitian chiral $p$-wave superfluids. The effective Hamiltonian for higher partial-wave pairing can be derived in the same way by considering higher channels. We have demonstrated that, similar to the NH $s$-wave superfluid, the non-Hermitian chiral $p$-wave superfluids also have a reentrant SF transition, and the critical points of the phase transition are characterized by pure imaginary gaps, which correspond to the emergence of EPs. 

Taking the $p+ip$ SF as an example, some special and characteristic behaviors of the NH pairing SFs are revealed, like the upper bound of $\Im E/\Re E$ and the folding-like shape energy spectrum. Another notable feature observed from the energy spectrum is the robustness of point nodes and line nodes under dissipation, which might suggest dissipation can be used as a probe, in unconventional superconductors, for detecting the information of the structure of gap nodes~\cite{ando2015topological,li2019topoSC}, besides angle-resolved photoemission spectroscopy (ARPES)~\cite{Kushnirenko2018ARPES,OKAZAKI2013ARPES,gao2024arpes} and measuring thermodynamic quantities~\cite{chapai2023PdTe,MATSUDA2003,Hasan2022,imajo2016}.

We studied the energy spectra of the non-Hermitian $p+ip$ superfluid on a 2D square lattice to investigate whether dissipation induced by the NH pairing influences the topological properties of this system. We find that the edge state, in the cylinder geometry, is robust to the dissipation. Then, although the PBC energy spectrum covers a finite area on the complex plane, implying the universal skin effect can exist in a generic 2D geometry, we proved the absence of NHSE in the cylinder geometry, and further, the absence of NHSE in the square geometry under fully OBC. This explains why topological properties, like Chern number and the edge states, are robust to dissipation in the system we are concerned with.

Finally, we put forward Theorem~\ref{thm.theorem} and Corollary~\ref{corollary} as criteria by only considering the algebra form of the Hamiltonian to determine zero winding number in some models, so as to exclude the presence of NHSE, instead of specifically calculating the winding number of (quasi-) 1D systems. They also hopefully shed light on the discovery of GDSE in some higher-dimensional NH systems. A future study may reveal the behavior of skin modes in those systems with different algebra structures of the Hamiltonian in a similar method.

\section{Acknowledgements}
We are grateful to Yan He, Wen Huang, and Yu-Min Hu for many valuable discussions. This work is supported by the NSFC under Grant Nos. 12174273 and 11704267. 

\appendix 
\begin{widetext}
\section{Non-Hermitian mean-field theory by path integral approach}
\label{appA}
We present the details of the NH gap equation in this section. In the path integral representation, the partition function is written as :
\begin{eqnarray}
\mathcal Z & = & \int\mathcal{D}[\bar{c}, c]e^{-S(\bar{c},c)},\\
S(\bar{c},c) & = & \int_{0}^{\beta}d\tau \Big[\sum_{\bk,\sigma}\bar{c}_{\bk,\sigma}(\tau)(\partial_{\tau}+\xi_{\bk})c_{\bk,\sigma}(\tau)+\frac{1}{N}\sum_{\bk, {\bk}^\prime} V(\bk, {\bk}^\prime) \bar{c}_{\bk,\ua}(\tau)\bar{c}_{-\bk,\da}(\tau)c_{-\bk^\prime,\da}(\tau)c_{\bk^\prime,\ua}(\tau)\Big],
\end{eqnarray}
For the $p+ip$ chiral superfluid, we only consider the component with $l=1$ and $m=1$ in Eq.~\eqref{eq:Vkk'}, then the partition function becomes,
\begin{equation}
\mathcal Z= \int_{0}^{\beta}d\tau \Big[\sum_{\bk,\sigma}\bar{c}_{\bk,\sigma}(\tau)(\partial_{\tau}+\xi_{\bk})c_{\bk,\sigma}(\tau)+  \frac{V_1}{N} \sum_{\bk, \bkp}\Gamma_{\bk} \Gamma_{\bkp}^* \bar{c}_{\bk,\ua}(\tau)\bar{c}_{-\bk,\da}(\tau)c_{-\bkp,\da}(\tau)c_{\bkp,\ua}(\tau)\Big],
\end{equation}
where $c$ and $\bar{c}$ are Grassmann variables, and $V(\bk,\bkp)$ is expanded as mentioned in the main text,  with $\Gamma_{\bk}=\sqrt{4\pi}Y_{1,1}(\theta_{\bk},\phi_{\bk})$, and $\Gamma_{\bk}^*=\sqrt{4\pi}Y_{1,1}^*(\theta_{\bk},\phi_{\bk})$.  As standard Hubbard-Stratonovich transformation, after introducing auxiliary bosonic fields $\bar{\Delta}_{\bk}(\tau)$ and $\Delta_{\bk}(\tau)$, the partition function is obtained, 
\begin{eqnarray}
\mathcal Z & = & \int\mathcal{D}[\bar{\Delta},\Delta,\bar{c}, c]e^{-S(\bar{\Delta},\Delta,\bar{c},c)},\\
S(\bar{\Delta},\Delta,\bar{c},c)&=&\int_{0}^{\beta}d\tau\Big[\sum_{\bk,\sigma}\bar{c}_{\bk,\sigma}(\tau)(\partial_{\tau}+\xi_{\bk})c_{\bk,\sigma}(\tau)+\sum_{\bk}\Big(\Gamma^*_{\bk}\bar{\Delta}_{\bk}(\tau)c_{-\bk,\da}(\tau)c_{\bk,\ua}(\tau) \nonumber \\
&+&\Gamma_{\bk}\Delta_{\bk}(\tau)\bar{c}_{\bk,\ua}(\tau) \bar{c}_{-\bk,\da}(\tau)+\frac{\bar{\Delta}_{\bk}(\tau)\Delta_{\bk}(\tau)}{V_1/N}\Big)\Big].
\end{eqnarray}
By the Fourier transformation, $c_{\bk,\sigma}(\tau)=\sum_{\omega_n}e^{-i\omega_n\tau}c_{\bk,\sigma}(\omega_n)/\sqrt{\beta}$, $\bar{c}_{\bk,\sigma}(\tau)=\sum_{\omega_n}e^{i\omega_n\tau}\bar{c}_{\bk,\sigma}(\omega_n)/\sqrt{\beta}$, $\Delta_{\bk,\sigma}(\tau)=\sum_{\Omega_l}e^{-i\Omega_l\tau}\Delta_{\bk,\sigma}(\Omega_l)/\sqrt{\beta}$, and $\bar{\Delta}_{\bk,\sigma}(\tau)=\sum_{\Omega_l}e^{i\Omega_l\tau}\bar{\Delta}_{\bk,\sigma}(\Omega_l)/\sqrt{\beta}$, we obtain
\begin{eqnarray}
S(\bar{\Delta},\Delta,\bar{c},c)&=&\frac{N}{V_1}\sum_{\bk,\Omega_l}\bar{\Delta}_{\bk}(\Omega_l)\Delta_{\bk}(\Omega_l)+\sum_{\bk,\sigma,\omega_n}\bar{c}_{\bk,\sigma}(\omega_n)(-i\omega_n+\xi_{\bk})c_{\bk,\sigma}(\omega_n)\\ \nonumber
&+&\sum_{\bk,\omega_n,\Omega_l}\frac{1}{\sqrt{\beta}} \Gamma^*_{\bk}\bar{\Delta}_{\bk}(\Omega_l)c_{-\bk,\da}(\Omega_l-\omega_n)c_{\bk,\ua}(\omega_n)\\ \nonumber
&+&\sum_{\bk,\omega_n,\Omega_l}\frac{1}{\sqrt{\beta}}\Gamma_{\bk}\Delta_{\bk}(\Omega_l)\bar{c}_{\bk,\ua}(\omega_n) \bar{c}_{-\bk,\da}(\Omega_l-\omega_n),\
\end{eqnarray}
where $\omega_n$  ($\Omega_l$) is the Matsubara frequencies of fermions (bosons). In the mean-field theory by saddle-point approximation, we can assume that in the superfluid state the pairs of bosons are condensed, and the dominant contribution comes from the mode at $\Omega_l=0$ and $\bk=0$, so we neglect the spatial and temporal fluctuations of $\Delta_{\bk}(\Omega_l)$ and $\bar{\Delta}_{\bk}(\Omega_l)$. Then the action is simplified as 
\begin{equation}
S(\bar{\Delta},\Delta,\bar{c},c)=\frac{\beta N}{V_1} \bar{\Delta}\Delta+\sum_{\bk,\omega_n} 
\begin{pmatrix}
\bar{c}_{\bk,\ua}(\omega_n) & c_{-\bk,\da}(-\omega_n)
\end{pmatrix}
\begin{pmatrix}
-i\omega_n+\xi_{\bk} & \Delta_{\bk}\\
\bar{\Delta}_{\bk} & -i\omega_n-\xi_{\bk}) \\
\end{pmatrix}
\begin{pmatrix}
c_{\bk,\ua}(\omega_n) \\ \bar{c}_{-\bk,\da}(-\omega_n)
\end{pmatrix}
\end{equation}
where $\Delta=\Delta_{\bk=0}(0)/\sqrt{\beta}$, $\Delta_{\bk}=\Gamma_{\bk}\Delta_{\bk=0}(0)/\sqrt{\beta}=\Gamma_{\bk}\Delta$, and $\bar{\Delta}_{\bk}=\Gamma^*_{\bk}\bar{\Delta}_{\bk=0}(0)/\sqrt{\beta}=\Gamma^*_{\bk}\bar{\Delta}$.
After integrating out the fermionic degrees of freedom by the Gaussian integral identity, we obtain
\begin{equation}
\mathcal Z  =  \int\mathcal{D}[\bar{\Delta}, \Delta]e^{-S_{\text{eff}}(\bar{\Delta},\Delta)},
\end{equation}
where the effective action is derived as
\begin{eqnarray}
S_{\text{eff}}(\bar{\Delta},\Delta) & = & \frac{\beta N}{V_1} \bar{\Delta}\Delta-\sum_{\bk,\omega_n}\ln\Big[-\det
\begin{pmatrix}
-i\omega_n+\xi_{\bk} & \Delta_{\bk}\\
\bar{\Delta}_{\bk} & -i\omega_n-\xi_{\bk}) \\
\end{pmatrix}
\Big],\\ \nonumber
&=&\frac{\beta N}{V_1} \bar{\Delta}\Delta-\sum_{\bk,\omega_n} \ln (\omega_n^2+\xi_{\bk}^2+\bar{\Delta}_{\bk}\Delta_{\bk}).
\end{eqnarray}

Employing the saddle-point approximation $\partial S_{\text{eff}}/\partial \Delta= \partial S_{\text{eff}}/\partial \bar{\Delta}=0$, we obtain the gap equation for the NH chiral $p$-wave SFs,
\begin{equation}
\sum_{\bk}\frac{|\Gamma_{\bk}|^2}{2E_{\bk}}\tanh \frac{\beta E_{\bk}}{2}=\frac{N}{V_{1}},
\end{equation}
where $E_{\bk}=\sqrt{\xi_{\bk}^{2}+\Delta_{\bk}\bar{\Delta}_{\bk}}$.
If we take $l=0$ in Eq.~\eqref{eq:Vkk'}, the gap equation is reduced to $\sum_{\bk}\tanh\big(\beta\sqrt{\xi^{2}_{\bk}+\bar{\Delta}\Delta}/2\big)/\big(2\sqrt{\xi^{2}_{\bk}+\bar{\Delta}\Delta}\big)=N/V_1$, which is the NH gap equation for the $s$-wave superfluid. 

\section{Condensation energy of the superfluid phase}
\label{appB}

Condensation energy is defined as the energy difference between the superfluid state and the normal state.
\begin{eqnarray}
E_{\text{c}}&=&\frac{1}{\beta}\Big(S_{\text{eff}}(\Delta,\bar{\Delta})-S_{\text{eff}}(0,0)\Big)
=\frac{N}{V_{1}}\bar{\Delta}\Delta-\frac{1}{\beta}\sum_{\bk,\omega_n}\Big[ \ln \Big(\omega_n^2+\xi_{\bk}^2+\bar{\Delta}_{\bk}\Delta_{\bk}\Big)-\ln \Big(\omega_n^2+\xi_{\bk}^2\Big)\Big] \\
&=&\frac{N}{V_{1}}\bar{\Delta}\Delta-\frac{1}{\beta}\sum_{\bk,\omega_n}\ln \Big(1+\frac{\bar{\Delta}_{\bk}\Delta_{\bk}}{\omega_n^2+\xi_{\bk}^2}\Big).
\end{eqnarray}

By using the Matsubara frequency summation technique, the condensation energy is given as:
\be
E_{\text{c}}=\frac{N}{V_{1}}\bar{\Delta}\Delta-\sum_{\bk}\Big(\sqrt{\xi_{\bk}^{2}+\Delta_{\bk}\bar{\Delta}_{\bk}}-\lvert \xi_{\bk} \rvert \Big).
\label{eq:Ec}
\ee

\section{Explanation of some characteristics of the phase boundary}
\label{appC}

In this section, we study the physical origin of the phase boundary between the normal phase and the metastable superfluid phase, which is associated with the pure imaginary superfluid order parameters. 

A distinctive feature of non-Hermitian systems is the emergence of exceptional points (EPs), where the Hamiltonian cannot be diagonalized~\cite{Heiss2012,berry2004,Heiss2004}, and the eigenstates and eigenvalues are degenerate~\cite{Bergholtz2021RMP}. To investigate the conditions of EPs, we consider a general two-level system:
\begin{equation}
H=(d_{0R}+id_{0I})\sigma_0+(d_{1R}+id_{1I})\sigma_x+(d_{2R}+id_{2I})\sigma_y+(d_{3R}+id_{3I})\sigma_z,
\label{eq:2levelH}
\end{equation}
where $d_{iR(I)}$ is real, $\sigma_0$ is the identity matrix, and $\sigma_{x,y,z}$ represent Pauli matrices. When $d_{iI}=0, (i=0,1,2, 3)$, the Hamiltonian is reduced to the Hermitian one.
The eigenvalues of Eq.~\eqref{eq:2levelH} are $E_{\pm}=(d_{0R}+id_{0I})\pm\sqrt{\sum_{i=1}^3(d_{iR}+id_{iI})^2}$, with energy difference as $\Delta E=E_{+}-E_-=2\sqrt{\sum_{i=1}^3(d_{iR}+id_{iI})^2}$.
For degenerate eigenstates, let $\Delta E$ be zero, and one can find the conditions for EPs:
\bea
(d^2_{1R}+d^2_{2R}+d^2_{2R})-(d^2_{1I}+d^2_{2I}+d^2_{3I})&=&0,\\
d_{1R}d_{1I}+d_{2R}d_{2I}+d_{3R}d_{3I}&=&0. 
\eea
Notably, these two necessary conditions for EPs are general for any Hamiltonian without any symmetry constraint. For non-Hermitian systems with parity-time ($PT$) or parity-particle-hole ($CP$) symmetry, only a single condition is needed~\cite{Okugawa2019}.

The non-Hermitian chiral $p$-wave superfluids on a two-dimensional square lattice can be written in Pauli matrix representation: 
\be 
H(\bk)=\xi_{\bk}\sigma_z+\Delta \sin k_x\sigma_x-\Delta \sin k_y\sigma_y,
\label{eq:Hk}
\ee 
where the order parameter is assumed to be complex as $\Delta=a+ib, (b\ne 0)$, and $\mu=-1$. Here we have ignored the factor $\sqrt{3/2}$ in $\Delta_{\bk}$ and $\bar{\Delta}_{\bk}$ without losing generality. Imposing the above two conditions on $H(\bk)$ in Eq.~\eqref{eq:Hk}, we obtain the following conditions of EPs for the chiral $p$-wave superfluids,
\bea
\xi^2_{\bk}+(a^2-b^2)(\sin^2 k_x+\sin^2 k_y)&=&0, \\
ab(\sin^2 k_x+\sin^2 k_y)&=&0.
\eea
Since $b\ne 0$, above conditions are satisfied when $a=0$ and  $\xi_{\bk}=\pm b\sqrt{\sin^2 k_x+\sin^2 k_y}$. This proves that EPs, satisfying $\xi_{\bk}=\pm b\sqrt{\sin^2 k_x+\sin^2 k_y}$, emerge when order parameter is pure imaginary.

Substituting the pure imaginary $\Delta=ib$ into $H(\bk)$ in Eq.~\eqref{eq:Hk}, in which $\bar{\Delta}_{\bk}=-\Delta^*_{\bk}$ is satisfied,  we obtain the following Hamiltonian,
\be
H(\bk,\Delta=ib)=i b\sin k_x \sigma_x-i b \sin k_y\sigma_y+\xi_{\bk}\sigma_z.
\label{eq:ImH}
\ee
It can be verified that $(CP)H(\bk,\Delta=ib)(CP)^{-1}=-H(\bk,\Delta=ib)$ holds, where $CP=\sigma_x K$, $\sigma_x$ is Pauli matrix, and $K$ is complex conjugate operator. It means the non-Hermitian Hamiltonian with pure imaginary order parameter in Eq.~\eqref{eq:ImH} satisfies $CP$ symmetry.
$CP$ symmetry requires two eigenvalues appear in pure imaginary or anti-complex-conjugate pairs, as $E_{\bk}\in i \mathbb R$ or $(E_{\bk},-E_{\bk}^*)$.
When $\xi_{\bk}=\pm b\sqrt{\sin^2 k_x+\sin^2 k_y}$, $H(\bk,\Delta=ib)$ in Eq.~\eqref{eq:ImH} cannot be diagonalized, and the eigenstates are degenerate, indicating the spontaneous $CP$ symmetry breaking, accompanied by emergence of the EPs.
Furthermore, for a $d$-dimensional system with $CP$ symmetry, its EPs will construct $(d-1)$-dimensional surface~\cite{Yoshida2019,Budich2019,yamamoto2019theory}. For our two-dimensional system, the EPs exist in the form of exceptional lines (ELs). As shown in Fig~\ref{fig:fig3} (c) in the main text, the two areas separated by the ELs correspond to the $CP$ symmetry phase (white area) and the $CP$ symmetry spontaneous breaking phase (blue area), respectively.

If we consider the gap equation written in spherical harmonic form:

\begin{equation}
\sum_{k}\frac{4\pi Y_{1,1}(\theta_k,\phi_k)Y^{*}_{1,1}(\theta_k,\phi_k)}{2E_{k}}=\frac{N}{V_{1}},
\label{gap:spher}
\end{equation}

with $E_{k}=\sqrt{\xi_{k}^{2}+\frac{3}{2}\Delta^{2}|Y_{1,1}(\theta_k,\phi_k)|^2}$ and $V_1=U_1+i\gamma/2$. For $N\to\infty$, we have:
\begin{equation}
\begin{aligned}
\frac{1}{V_1}&=\int_{-1}^{1}\mathrm{d} \cos\theta_k\int_{-\omega_0}^{\omega_0} \mathrm{d}\xi N(\xi)\frac{3(1-\cos^2\theta_k)}{8\sqrt{\xi^{2}+\frac{3}{2}\Delta^{2}\sin^2\theta_k}}\\
&=\rho_0\frac{3}{4}\int_{-1}^{1}\mathrm{d}x(1-x^2)\ln(\frac{\omega_0+\sqrt{\omega^2_0+\Delta^2_k}}{\sqrt{\frac{3}{2}}\Delta\sqrt{1-x^2}})\\
&=\rho_0\frac{3}{4}\int_{-1}^{1}\mathrm{d}x(1-x^2)\ln(\frac{2\omega_0}{\sqrt{\frac{3}{2}}\Delta\sqrt{1-x^2}}) \\
&=\rho_0\ln(\frac{e^{5/6}}{\sqrt{6}}\frac{2\omega_0}{\Delta}).
\label{cDOS}
           \end{aligned}
\end{equation}

We apply the constant DOS ($N(\xi)=\rho_0$) and set $\hbar$ to 1 during the calculation. Here, for the lpenultimate “=” of Eq. \eqref{cDOS}, we follow the BCS weak coupling approximation condition ($|\omega_0/\Delta|\gg 1$). Then we rewrite Eq. \eqref{cDOS} into this form:
\begin{equation}
\frac{1}{\rho_0}(\frac{U_1}{|V_1|^2}-i\frac{\gamma}{2|V_1|^2})=\ln(\frac{e^{5/6}}{\sqrt{6}}|\frac{2\omega_0}{\Delta}|)+i\mathrm{Arg}(\frac{e^{5/6}}{\sqrt{6}}\frac{2\omega_0}{\Delta}).
\label{11}
\end{equation}

Using the same way as that in the Ref.~\cite{yamamoto2019theory}, the analytic solution for the phase boundary of a semicircular shape can be derived. The results are quantitatively in agreement with the NH $s$-wave superfluid. However, one should also notice that the weak coupling condition will no longer be valid for a very large $|\Delta|$ on the boundary, which satisfies that the actual phase boundary is not a semicircle for large $\gamma$.

\section{A proof of an upper bound of the ratio $\Im E/\Re E$}
\label{appD}
In this section, we prove that, for the two-dimensional chiral superfluids on a square lattice with non-Hermitian pairing in Eq.~\eqref{eq:H_lattice} under the periodic boundary conditions, each eigenenergy of Eq.~\eqref{eq:BlochH} in the complex energy spectrum as shown in Fig.~\ref{fig:fig4} (c) satisfies following inequality:
\be
\Im E/\Re E \le \Im \Delta/ \Re\Delta.
\ee

We start with the Bloch Hamiltonian by ignoring the factor $\sqrt{3/2}$ in $\Delta_{\bk}$ and $\bar{\Delta}_{\bk}$ without losing generality, , 
\begin{equation}
H_{\bk}=
\begin{pmatrix}
\xi_{\bk} & \Delta_{\bk}\\
\bar{\Delta}_{\bk} & -\xi_{\bk} \\
\end{pmatrix}
=\begin{pmatrix}
\xi_{\bk} & \Delta (\sin{k_x}+i \sin{k_y})\\
\bar{\Delta}(\sin{k_x}-i \sin{k_y}) & -\xi_{\bk} \\
\end{pmatrix},
\label{eq:BlochH}
\end{equation}
where  $\xi_{\bk}=-2t(\cos{k_x}+\cos{k_y})-\mu$, and $\bar{\Delta}=\Delta$ by gauge fixing. The Bloch Hamiltonian can be diagonalized as 
$\Lambda_{H}=\begin{pmatrix}
E_{\bk} & 0 \\
0 &  -E_{\bk}\\
\end{pmatrix}$,
where $ E_{\bk}=\sqrt{\xi_{\bk}^{2}+\Delta_{\bk}\bar{\Delta}_{\bk}}=\sqrt{{\xi^{2}_{\bk}}+\Delta^{2}(\sin^2 k_x+ \sin^2 k_y)}$. To investigate the relation between $\Re E$
and $\Im E$, we introduce following symbols:  $\Delta=a+ib, E_{\bk}=x+iy$ with $a, b>0$, and rewrite the formulae of $E_{\bk}$, then we obtain two equations,
\bea
\begin{aligned}
x^2-y^{2}&=\xi^{2}_{\bk}+(a^2-b^2)(\sin^{2}k_x+\sin^{2}k_y), \\
 xy&=ab(\sin^{2}k_x+\sin^{2}k_y).
\end{aligned}
 \label{curveeq}
\eea
Once we fix the values of $(a,b,k_x,k_y)$, two equations above correspond to two set of hyperbolae about $x$ and $y$, and the intersection points of these curves 
yielding one solution $(\Re E_{\bk},\Im E_{\bk})$ on the complex energy spectrum at given $(k_x,k_y)$. 
With fixed $\Delta$, and let $k_x$ and $k_y$ spread the entire first Brillouin zone, it's clear that the intersection points are distributed on a surface rather than a curve, due to the $\bk$ dependence in the coefficients of hyperbola equations. This is consistent with the numerical results in Fig.~\ref{fig:fig4} in the main text.

Now, let us focus on the two equations in Eq.~\eqref{curveeq}. At point nodes, where $\sin^{2}k_x+\sin^{2}k_y=0$, 
we will get ($x=\pm\xi_{\bk},y=0$) on the $\Re E-\Im E$ spectrum, where $s:=y/x=0$. In the case where $\sin^{2}k_x+\sin^{2}k_y \ne 0$ , above equations can be divided into:
\be
\frac{1}{s}-s=\frac{\xi^{2}_{\bk}}{ab(\sin^{2}k_x+\sin^{2}k_y)} +\frac{a^2-b^2}{ab},
 \label{eq:slope1}
\ee
where $s=y/x$ is defined, which is the ratio of the imaginary part to the real part of each $E$. It's easy to find that the function $f(s):=1/s-s$ decreases with $s$ increasing in both regimes $s>0$ and $s<0$. This means that when the right-hand side of Eq.~\eqref{eq:slope1} or the function $f(s)$ arrives at the minimum value, the ratio $s$ takes the maximum value. Because $\xi^{2}_{\bk}/(\sin^2 k_x+\sin^2 k_y)\ge 0$ always holds, and $a,b>0$ is assumed, in the case  that $\sin^{2}k_x+\sin^{2}k_y \ne 0$, one can find $s_{\text{max}}=(y/x)_{\text{max}}=b/a$, implying
\be
\frac{\Im E}{\Re E}\le \frac{\Im \Delta}{\Re \Delta}.
\label{eq:ineq}
\ee
It means the ratio of the imaginary part to the real part of each energy value $E$ in the complex energy spectrum does not exceed the order parameter gap $\Delta$  counterpart.

A similar constraint exists in other systems with the same algebra structure, and we take the non-Hermitian $s$-wave superfluidity as an example. The shape of the energy spectrum for a non-Hermitian $s$-wave superconductor is determined by curve equations:
\begin{equation}
\begin{aligned}
x^2-y^{2}&=\xi^{2}_{\bk}+a^2-b^2, \\
 xy&=ab.
           \end{aligned}
\end{equation}
The second equation above is a fixed-coefficient hyperbola equation, which is independent of $(k_x,k_y)$, once $\Delta=a+ib$ is chosen. The energy spectrum will be limited on a curve, rather than an area spanned by groups of curves, when $k_x$ and $k_y$ traverse the entire Brillouin zone. Despite this, we can verify that the energy ratio constraint still exists. According to the theorem given by~\cite{Fang2022NatComm}, there is no universal skin effect for a non-Hermitian $s$-wave superconductor. 

Similar energy ratio constraints can be derived for the non-Hermitian higher-order superconductors, such as $d$-wave, in this approach. Since the coefficients of hyperbolic equations contain $\bk$ dependence introduced by $\Delta_{\bk}$, it directly leads that the $\Re E-\Im E$ spectrum appears in a shape with nonzero area. 

In addition, we comment that the folding-like behavior observed in Fig.~\ref{fig:fig4} (c) can be explained by Eq.~\eqref{eq:ineq}. As we discussed earlier, the function $f(s):=1/s-s$ decreases with increasing $s$ in both regimes $s>0$ and $s<0$. Because the coefficient $\xi^{2}_{\bk}/(\sin^{2}k_x+\sin^{2}k_y)\ge 0$ holds for any $\bk$, according to Eq.~\eqref{eq:slope1}, the ratio $s=y/x=\Im E/ \Re E$ arrives at its maximal value when $\xi_{\bk}$ approaches to zero. Let us consider the evolution of $\xi_{\bk}$ from negative to zero. In this process, the coefficient $\xi^{2}_{\bk}/(\sin^{2}k_x+\sin^{2}k_y)$ and the right hand side of Eq.~\eqref{eq:slope1} decrease, leading to the increasement of the ratio $\Im E/ \Re E$. The ratio arrives at its maximal value as $\Im \Delta/\Re \Delta$ when $\xi_{\bk}=0$. When $\xi_{\bk}$ increases further from zero to a positive value,  the coefficient $\xi^{2}_{\bk}/(\sin^{2}k_x+\sin^{2}k_y)$ and the right hand side of Eq.~\eqref{eq:slope1} increase, leading to the decrement of the ratio $\Im E/ \Re E$. Therefore, points with larger ratio than $\Im \Delta/\Re \Delta$ are prohibited in the complex energy spectrum. More specifically, we can fix the value of $k_x$, and let $k_y$ evolve from $-\pi$ to $\pi$, and then we obtain a curve on the complex energy plane, which could not exceed the function $y=(\Im \Delta/\Re \Delta) x$. When the curve ``touches" this boundary line, it reflects inwards. This explains the folding-like behavior intuitively.

Finally, as an example, we consider the energy spectrum of $d_{x^2-y^2}$ superconductivity~\cite{xiang2022d} on the square lattice with PBC to show the robustness of the line nodes, where $\Delta_k$ can be taken as $\Delta(\cos k_x-\cos k_y)/2$. The quasi-particle energy is given by,
\begin{equation}
E_k=\sqrt{\Big(-2t(\cos k_x+\cos k_y)-\mu\Big)^2+\frac{\Delta^2}{4}(\cos k_x-\cos k_y)^2},
\end{equation}
By considering the complex-valued $\Delta$, it is easy to verify that $\Im (E_k)=0$ along the line nodes $k_x=\pm k_y$ and $k_x=2\pi -k_y$, as shown in FIG.~\ref{fig:fig6}.

\begin{figure}[ht]
\centering
\includegraphics[width=10cm]{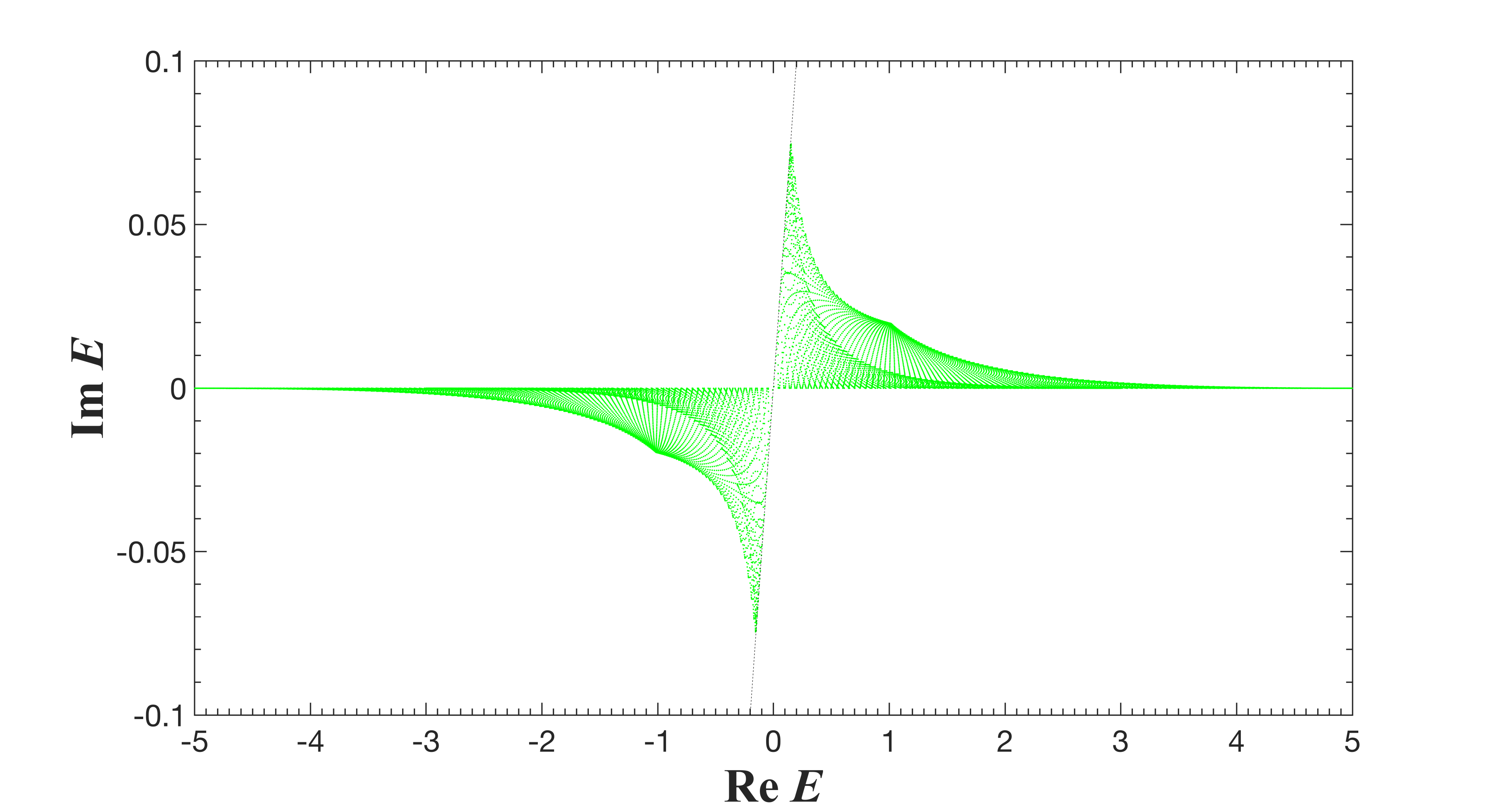}
\caption{The energy spectrum of non-Hermitian $d_{x^2-y^2}$ superconductivity. One can find the robustness of line nodes in this case, where $\mu/t=-1, \Delta/t=0.2+0.1i$, and the system size $N=250\times250$. The black dashed line represents the linear function $\Im E=(\Im\Delta/\Re\Delta)\Re E$.}
\label{fig:fig6}
\end{figure}

\section{The proof and applications of Theorem~\ref{thm.theorem}}
\label{appE}

\subsection{A proof of Theorem~\ref{thm.theorem}}
Firstly, we prove Theorem~\ref{thm.theorem} provided in the main text. We first confirm that the form of even function $F(k)=\det[H(k)-E_{0}\mathbb{I}]$ is related to a palindromic polynomial of $z$ after substitution $z=e^{ik}$. With the palindromic polynomial as a bridge, we can prove the complex integral in Eq.~\eqref{eq:WN} is zero, in the following two steps.

\begin{lemma}
If $F(k)\equiv \det[H(k)-E_{0}\mathbb{I}]=\sum_{l,m,n,a,b} c_{l}\cos^m(ak)\sin^n(bk)$ is even of $k$, then $F(k)$ can be transformed into $P_{2\alpha}(z)/z^\alpha$, where $P_{2\alpha}(z)$ is a complex palindromic polynomial with $\deg~P_{2\alpha}(z)=2\alpha$. Here,  $z=e^{ik}$ ($k\in\mathbb{R}$), $l,m,n,a,b\in\mathbb{N}$, $c_{l}\in\mathbb{C}$ and $\alpha=(ma+nb)_{\mbox{max}}$.
 \label{le:lemma1}
\end{lemma}

The function $P_{2\alpha}(z)$, as a complex palindromic polynomial of order $2\alpha$, can be written in following general form: $P(z)=t(z^{2\alpha} +k_1z^{2\alpha-1}+\cdots +k_{\alpha-1}z^{\alpha+1}+k_\alpha z^{\alpha}+k_{\alpha-1}z^{\alpha-1}+\dots+k_1z^1+1)$, where $t \in\mathbb{C}$ ($t\ne 0$), and $k_{i}\in\mathbb{C}, (i=1, 2, \dots ,\alpha-1, \alpha)$,  and $z$ is a complex variable.  The palindromic polynomial $P_{2\alpha}(z)$ satisfies the condition that the coefficients of $z^{i}$ and $z^{2\alpha-i}$ are the same. 

We first present a proof for Lemma~\ref{le:lemma1}.

\begin{proof}
Applying the substitution $z=e^{ik}$, the winding number in Eq.~\eqref{eq:WN} is given by: 
\begin{equation}
w_{\text{BZ},E_{0}}=\frac{1}{2\pi i}\int_{0}^{2\pi}\frac{{\rm d}}{{\rm d}k}\ln\Big[\det[H(k)-E_{0}\mathbb{I}]\Big]{\rm d}k=\frac{1}{2\pi i}\int_{0}^{2\pi}\frac{{\rm d}}{{\rm d}k}\ln F(k) {\rm d}k
=\frac{1}{2\pi i}\oint_{\left | z \right | =1}\frac{{\rm d}}{{\rm d}z}\ln f(z) {\rm d}z,
\end{equation}
where $f(z)$ is the result of $F(k)$ after the substitutions: $\sin k=(z^2-1)/(2iz)$ and $\cos k=(z^2+1)/(2z)$. 

Now, let us verify that any $\cos^m (ak)$ can be expressed by a palindromic polynomial. 
For any integer power $m$, $(ma\le\alpha)$, $\cos^m (ak)$ can be expressed by binomial theorem:
\begin{align}
\cos^m (ak)&=\Big(\frac{z^{2a}+1}{2z^a}\Big) ^m=\frac{\sum_{j=0}^{m} \binom{m}{j} z^{2a(m-j)}  }{2^m z^{am}}= \frac{z^{2am}+mz^{2a(m-1)}+\dots+mz^{2a}+1}{2^m z^{am}}. 
\label{eq:cos}
\end{align}
The numerator in Eq.~\eqref{eq:cos} is a palindromic polynomial of $z$, whose order is $2am$. 

A similar process can be applied to $\sin (bk)$ with even orders, and one can prove $\sin^{2n} (bk)$ can be expressed by a palindromic polynomial in the same way.
Besides, the product of two palindromic polynomials would also be a palindromic polynomial, so the product of $\cos^m(ak)$ with any \textit {integer} power $m$ and  $\sin^{2n}(bk)$ with \textit{even} powers is also equivalent to the form $P_{2\gamma}(z)/z^\gamma$. However, exceptions occur when containing $\sin k$ with odd powers as $\sin^{2n+1}(bk)$:
\begin{align}
\sin^{2n+1}(bk)&=\Big(\frac{z^{2b}-1}{2iz^b}\Big) ^{2n+1}= \frac{\sum_{j=0}^{2n+1}\binom{2n+1}{j} z^{2b(2n+1-j)}(-1)^j}{2^{2n+1}i^{2n+1}z^{b(2n+1)}}\nonumber \\
&= \frac{z^{2b(2n+1)}-(2n+1)z^{2b(2n)}+\dots+(2n+1)z^{2b}-1}{2^{2n+1}i^{2n+1}z^{b(2n+1)}} 
\end{align}
It is obvious that the introduction of coefficients $(-1)^j$ makes it not satisfy the condition of a palindromic polynomial. 

Therefore, if the polynomial $F(k)\equiv \det[H(k)-E_{0}\mathbb{I}]$ does not contain any odd power of $\sin(k)$, i.e. $F(k)$ is an even function of $k$,  then $F(k)$ can be transformed into the form: $f(z)=P_{2\alpha}(z)/z^\alpha$, where $P_{2\alpha}(z)$ is a palindromic polynomial.
The $2\alpha$ order of the polynomial $P_{2\alpha}(z)$ can be easily read out because $\sin k=(z^2-1)/(2iz)$ and $\cos k=(z^2+1)/(2z)$, whose orders in numerators are twice those in denominators. $\alpha$ can be understood in this way: all terms $\cos^m(ak)\sin^n(bk)$ in the $F(k)$ would transform into the form $P_{2(am+bn)}(z)/z^{am+bn}$. After the reduction of fractions to a common denominator, we get $f(z)$ as a single fraction $P_{2\alpha}(z)/z^\alpha$. Therefore, $\alpha$ takes the maximal value of $(am+bn)$. 

\end{proof}

We then show that $P(z)$ in such a form in Lemma~\ref{le:lemma1} leads to a trivial winding number:

\begin{lemma}
If $P(z)$ is a complex palindromic polynomial with $\deg~P(z)=2\alpha$, where $\alpha\in \mathbb{Z^+} $, then 
\begin{equation}
I=\oint_{\left | z \right | =1}\frac{{\rm d}}{{\rm d}z}\ln \Big[\frac{P(z)}{z^\alpha}\Big]{\rm d}z=0.
\end{equation}
\label{le:lemma2}
\end{lemma}

\begin{proof}
First, let us review the residue of the logarithmic derivative ${\rm d}\ln[f(z)]/{\rm d}z=f^{\prime}(z)/f(z)$ at the point $z=a$. If $a$ is a zero of $f(z)$ with multiplicity $n$, and if $b$ is a pole of $f(z)$ with order $m$, then the residues of $f^{\prime}(z)/f(z)$ at points $a$ and $b$ are $n$ and $-m$, respectively, where $a$ and $b$ must be the first order poles of $f^{\prime}(z)/f(z)$~\cite{Lav_complecbook}. By definition of a palindromic polynomial $P(z)$ with $\deg P(z)=2\alpha$, we can write $P(z)$ in following general form: $P(z)=t(z^{2\alpha} +k_1z^{2\alpha-1}+\cdots +k_{\alpha-1}z^{\alpha+1}+k_\alpha z^{\alpha}+k_{\alpha-1}z^{\alpha-1}+\dots+k_1z+1)$.
It is clear that $P(z)/z^\alpha$ has only one pole at $z=0$ with order $\alpha$, so $\Res\ln[P(z)/z^\alpha]|_{z=0}=-\alpha$.

Besides, we need to discuss the residue at zeros of $P(z)/z^\alpha$. It is obvious that $z=0$ is not the solution of $P(z)=0$. Hence, we can rewrite the polynomial $P(z)$ in a product form: $t(z-\beta_1)(z-\beta_2)\dots(z-\beta_{2\alpha})$, ($\beta_i\ne 0,i=1,2,\cdots,2\alpha$), in which $\beta_i$ is a zero of $P(z)$.
According to the property of real palindromic polynomial $f(x)$, if $x_0\ne 0$ is one zero of $f(x)$, then $1/x_0$ is also a zero of $f(x)$ with the same multiplicity~\cite{Sun2015}. This property can be generalized to a complex palindromic polynomial $P(z)$. So these $2\alpha$ zeros can be written in pairs: ${(\beta_1,\beta_2), (\beta_3, \beta_4), \dots, (\beta_{2\alpha-1}, \beta_{2\alpha})}$, each pair satisfies $\beta_{2i-1}\beta_{2i}=1, (i=1, \dots, \alpha)$. 
Therefore, the distribution of these zeros on the complex plane can be divided into two groups: a pair of zeros on the circle $|z|=1$ at the same time ($|\beta_{2i-1}|=|\beta_{2i}|=1$), one zero inside the circle while the other outside the circle ($|\beta_{2i-1}|<1, |\beta_{2i}|>1$).

\begin{figure}[ht]
\centering
\includegraphics[width=7cm]{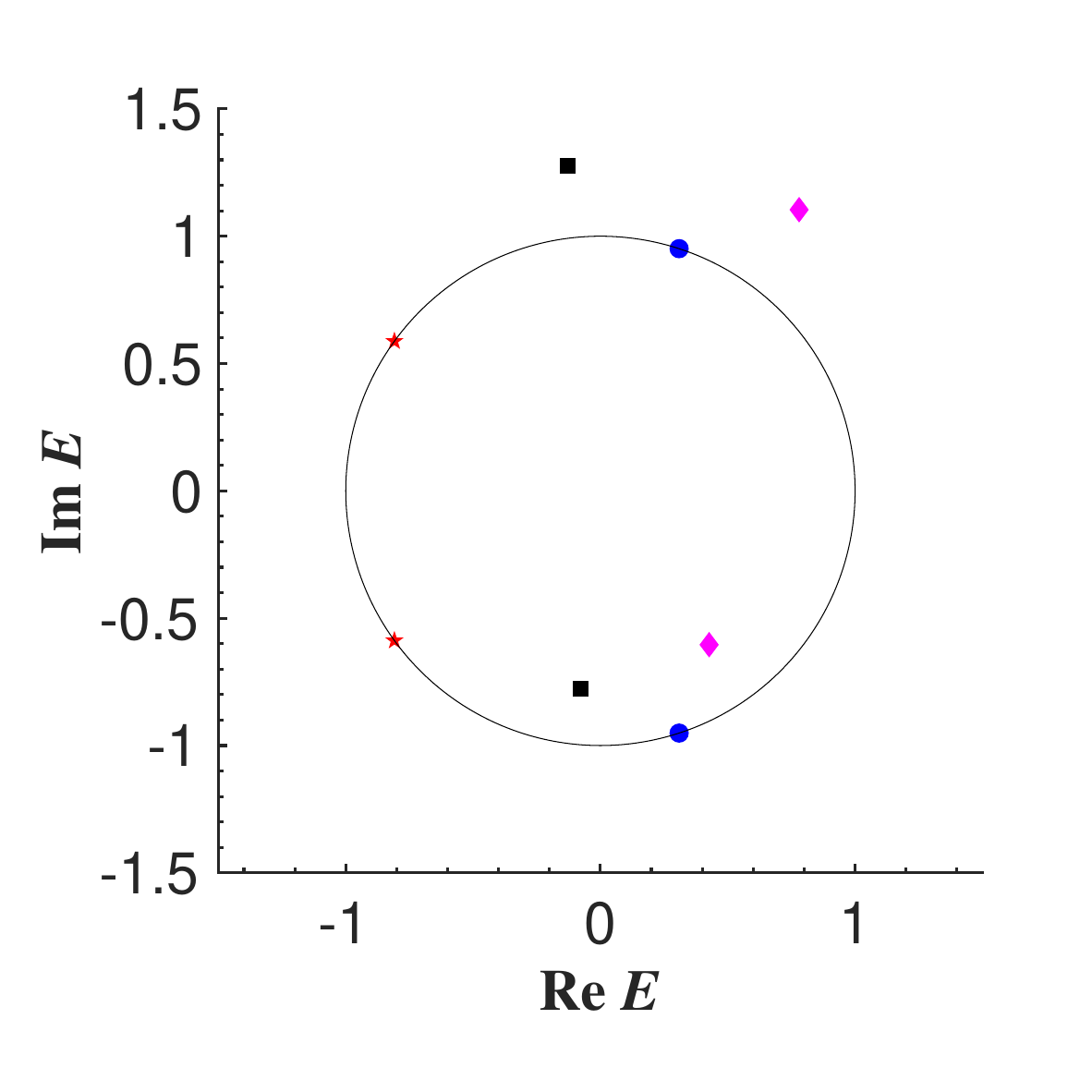}
\caption{An illustration of the distribution of zeros of palindromic polynomial $(1+i)z^8+(1-i)z^7+(2+i)z^6+(2-i)z^5+3z^4+(2-i)z^3+(2+i)z^2+(1-i)z+(1+i)$ as an example. Each pair of points with the same color and shape corresponds to a pair of zeros $(\beta_i,\beta_{i+1})$.}
\label{fig:fig7}
\end{figure}

We can calculate integral $I=\oint_{|z|=1}{\rm d}/{\rm d}z(\ln[P(z)/z^\alpha]){\rm d}z$ by the Residue theorem. Let us assume there are $j$ pairs of zeros on the circle $|z|=1$, where $0\le j\le \alpha$. Then there are $\alpha-j$ zeros inside the circle. Here we count all zeros and poles weighted by respective orders. Then the integral can be evaluated by,
\be
I=\oint_{|z|=1}\frac{{\rm d}}{{\rm d}z}\ln\Big[\frac{P(z)}{z^\alpha}\Big]{\rm d}z=2\pi i (\alpha-j) +\pi i \times 2j -2\pi i \alpha=0. 
\ee
These three terms above describe the contributions from the zeros of $P(z)/z^{\alpha}$ inside the circle, the contribution from the zeros on the circle,  and the contribution from the poles inside the circle, respectively. 
\end{proof}

From Lemma~\ref{le:lemma1} and Lemma~\ref{le:lemma2}, Theorem~\ref{thm.theorem} in the main text is proved.

\subsection{A simplified proof of Theorem~\ref{thm.theorem}}
The above process suggests the way to calculate the winding number by directly studying the algebraic structure of a Hamiltonian. Here we also provide a simplified version of the proof.
By applying the transformation $z=e^{ik}$, and taking the even function condition of $F(k)$ into consideration, the function $F(k)={\sum_{l,m,n,a,b}} c_{l}\cos^m(ak)\sin^n(bk)$ is transformed into $f(z)$,
\be
f(z)={\sum_{l,m,p,a,b}} c_{l}\Big(\frac{z^{2a}+1}{2z^a}\Big)^m\Big(\frac{z^{2b}-1}{2iz^b}\Big)^{2p} 
={\sum_{l,m,p,a,b}} c_{l}\Big(\frac{1+z^{2a}}{2z^a}\Big)^m\Big(\frac{1-z^{2b}}{2iz^b}\Big)^{2p} =f(1/z),
\label{eq:zeros}
\ee
where $n$ is a non-negative even number, $n=2p~(p\in\mathbb{N})$. 

Equation~\eqref{eq:zeros} implies that if $z_1$ is a zero of $f(z)$, then $z_2=1/z_1$ is also a zero of $f(z)$. In other words, the zeros of $f(z)$ always come in pairs $(z_1,z_2)$, satisfying $z_1z_2=1$.
Then the above discussion about the distribution of zeros and the residue of the logarithmic derivative can be applied to $f(z)$ again, and we have
\be
I=\oint_{|z|=1}\frac{{\rm d}}{{\rm d}z}\ln f(z){\rm d}z=0. 
\ee

\subsection{Applications of Theorem~\ref{thm.theorem} to two-band models}

In this section, we present two examples as applications of Theorem~\ref{thm.theorem} to two-band models for the trivial winding numbers.

\subsubsection{A non-Hermitian Kitaev chain Hamiltonian}
Here, let us consider a non-Hermitian Kitaev chain Hamiltonian by introducing the complex next-nearest neighbour hopping terms in real space:
\be
H=\mu c_m^\dagger c_m + \Big(t_1c_{m+1}^\dagger c_m+\Delta c_{m+1}^\dagger c_m^\dagger +h.c.\Big)+(t_2+i\gamma) \Big(c_{m+2}^\dagger c_m+c_m^\dagger c_{m+2}\Big),
\ee
where $\mu, t_1, t_2, \Delta, \gamma \in \mathbb{R}$ are the chemical potential, two hopping amplitudes, the $p$-wave pairing gap, and the degree of non-Hermiticity, respectively. In momentum space, the Hamiltonian matrix in the Pauli matrices basis is given by
\be
H(k)=\Big(\frac{\mu}{2}+t_1\cos k +(t_2+i\gamma)\cos (2k)\Big)\sigma_z-\big(\Delta\sin k\big)\sigma_y.
\ee
One can notice the characteristic equation of this Hamiltonian is $\det [H(k)-E]=E^2-\Big(\mu/2+t_1 \cos k+(t_2+i\gamma)\cos(2k)\Big)^2-\Delta^2\sin^2(k)=0$, which can be also viewed as a characteristic equation for a single-band model and base energy $E^2$, and the solutions of this single-band characteristic equation will give the eigenvalue $E^2$, i.e. the eigenvalue pairs $(E,-E)$. Considering the two-band Hamiltonian belongs to class D, due to the particle-hole symmetry $PH^T(k)P^{-1}=-H(-k)$, where $P=\sigma_x$~\cite{Kawabata2019PRX}, the eigenvalues must appear as the form of pairs $(E,-E)$. Hence, we develop the equivalent characteristic equation, i.e. the same GBZ conditions between the two-band Hamiltonian and one single-band Hamiltonian, and the winding number here can be used to characterize the information of the skin effect~\cite{Fang2020PRL}. According to Corollary~\ref{corollary}, for any reference energy $E_0\in\mathbb{C}$, $w(E_0)=0$, because $a_0(k)=0$ and $\det [H(k)]=-\Big(\mu/2+t_1 \cos k+(t_2+i\gamma)\cos(2k)\Big)^2-\Delta^2\sin^2(k)$ are both even functions of $k$, implying the absence of skin effect. We present the numerical results in Fig.~\ref{fig:fig8}, where the energy spectra with OBC and PBC overlap with each other, except for the robust Majorana zero modes with OBC in the presence of a line gap, and we confirm that there is no NHSE numerically.

\subsubsection{The non-Hermitian chiral $p+ip$ superfluid}
In fact, the similar process could be applied to the Hamiltonian~Eq.\eqref{eq:BlochH}. For researching the skin effect in a cylinder geometry with $x$-direction open boundary, we can fix one $k_y=k_{y0}$, and calculate the winding number for this PBC Hamiltonian $H(k_x, k_{y0})$. Since eigenvalues of this quasi-one-dimensional two-band Hamiltonian always appear as pairs $(E,-E)$,we can find an equivalent single-band characteristic equation for this model, and further, the winding number must characterize the skin effect. One can also check $\det [H(k)]$ is an even function of $k_y$ or $k_x$, for a fixed $k_{x0}$ or $k_{y0}$, and we get zero winding number in the specific open $y$- or $x$-boundary cylinder geometry. Therefore, the discussion based on the winding number in the main text is valid.

\begin{figure}[ht]
\centering
\includegraphics[width=10cm]{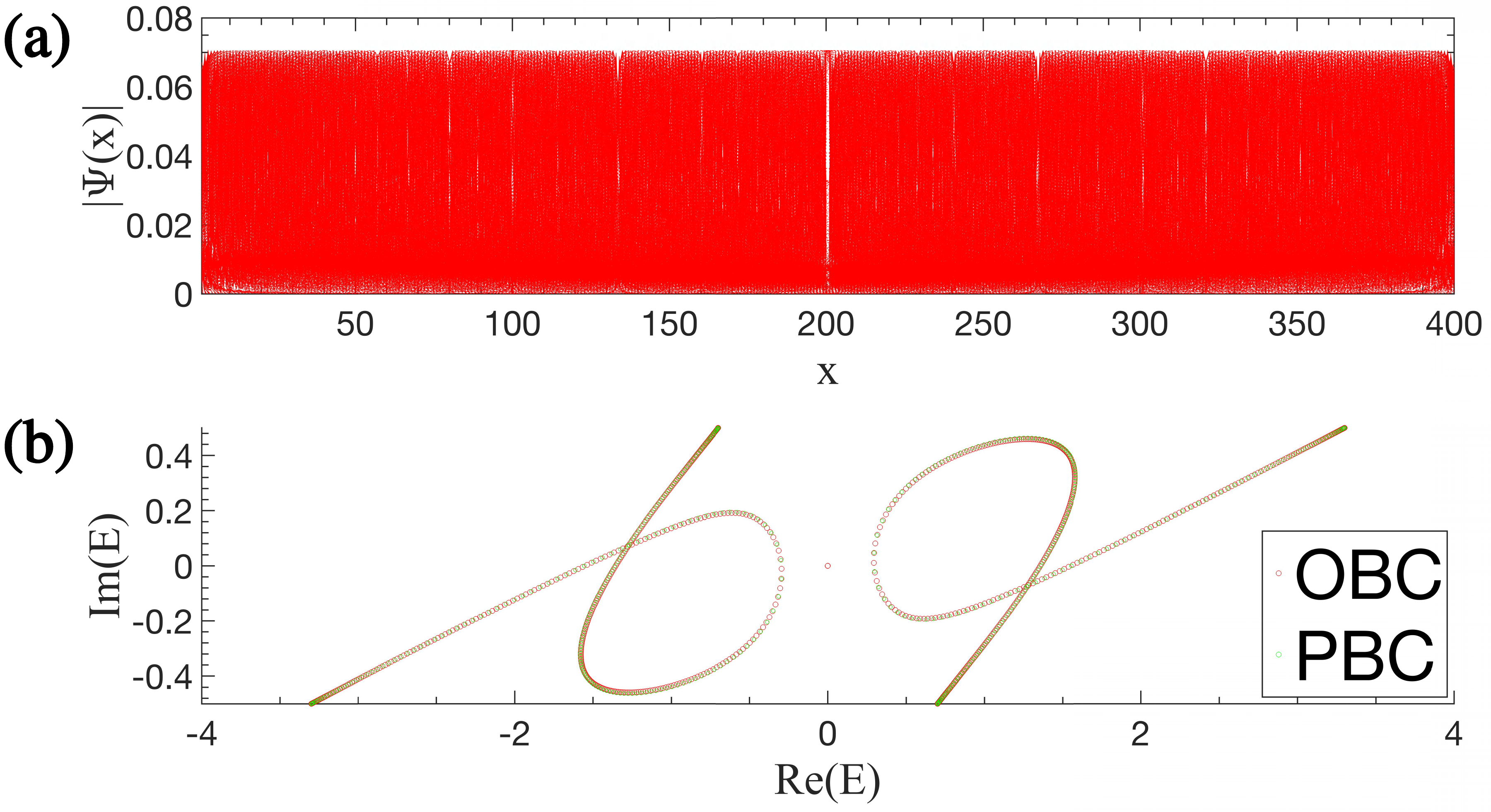}
\caption{(a) The distribution of wavefunctions of all bulk modes.  (b)The energy spectra with PBC and OBC. Parameters: $\mu=0.2$, $t_1=2$, $t_2=1.2$, $\gamma=0.5$, $\Delta=0.5$, and systems size $N=400$.}
\label{fig:fig8}
\end{figure}

\subsection{Application of Theorem~\ref{thm.theorem} to higher-dimensional systems}
Finally, we discuss the significance of Theorem~\ref{thm.theorem} for models in dimensions higher than one. We start from the 1D winding number:
\begin{equation}
w_{\text{BZ},E_0}=\frac{1}{2\pi i}\int_{-\pi}^{\pi}{\rm d}k\frac{\rm{d}}{{\rm d}k}\ln[\det(H(k)-E_{0}\mathbb{I})].
\label{eq:integral}
\end{equation}
When $F(k)=\det[H(k)-E_{0}\mathbb{I}]$ is an even function about $k$, we then have ${\rm d}\ln[\det(H(k)-E_{0}\mathbb{I})]/{\rm d}k={F}'(k)/F(k)$, which is obviously an odd function. We can conclude $w_{\text{BZ},E_0}=0$, because $k \in [-\pi,\pi]$. However, we should notice that the above conclusion is only valid when $F^{\prime}(k)/F(k)$ is well defined. This requirement is naturally satisfied when we consider 1D system, because the reference energy $E_0$ usually does not belong to the energy spectrum of $H(k)$, or $E_0$ is not an eigenvalue of $H(k)$. In other words, $F(k)\ne 0$. However, when we consider systems in higher dimensions, spectra with nonzero area may appear, and we need to determine whether GDSE exists in such a case~\cite{Fang2022NatComm}. Then we need to calculate the (quasi-) 1D spectral winding number for each straight line and all momenta in the entire BZ, which means some corresponding reference energies may be in the mapping energy spectrum of the straight line, and we have to calculate the winding number about some $E_0$ with $\det[H(k)-E_{0}\mathbb{I}]=0$. Hence, we need to cope with the possible divergency for some $k$ in integral~Eq.\eqref{eq:integral} now. Theorem~\ref{thm.theorem} gives a ``no-go'' theorem, in which the singularity has been taken into account during the derivation. Let us recall the possible distributions of roots of $P(z)$, those pairs of roots on the unit circle exactly correspond to the solutions of $\det[H(k)-E_{0}\mathbb{I}]=0$, which means there exist $k\in [0, 2\pi]$ such that $\det[H(k)-E_{0}\mathbb{I}]=0$. Hence, Theorem~\ref{thm.theorem} allows us to include the singularities and calculate the winding number in higher-dimensional systems.

We also need to emphasize that Theorem~\ref{thm.theorem} and Corollary~\ref{corollary} provided in the main text are only ready for obtaining the trivial winding number, and they may not guarantee the absence of the skin effect. For example, although we have ensured the universal trivial winding number for systems with reflection symmetry, the skin effect may also appear for some models, like the model of class AI~\cite{Liu2019}, where a different topological invariant, $\mathbb{Z}_2$ invariant, characterizes topological properties. Hence, although Theorem~\ref{thm.theorem} and Corollary~\ref{corollary} give us convenient way to obtain a zero winding number, when one aims to research the skin effect, they are only applicable to the case where the winding number can be used to characterize the skin effect exactly, such as the one-dimensional one-band model.

\end{widetext}


%

\end{document}